\documentclass[12pt,fleqn,reqno]{amsart}
\usepackage[numbers]{natbib}
\usepackage{hyperref}

\usepackage{amsfonts,amssymb,amsthm,amsmath}
\usepackage{epsfig}
\usepackage{graphicx}
\usepackage{mathrsfs}
\usepackage{color}

\usepackage{todonotes}
\usepackage{soul}

\usepackage[margin=1.37in]{geometry}

\theoremstyle{plain}
\newtheorem{theorem}{Theorem}
\newtheorem{proposition}[theorem]{Proposition}
\newtheorem{lemma}[theorem]{Lemma}
\newtheorem{corollary}[theorem]{Corollary}

\newtheorem{rem}[theorem]{Remark}

\numberwithin{theorem}{section}
\numberwithin{equation}{section}

\newcommand{\nc}{\newcommand}

\nc{\be}{\begin{equation}}
\nc{\la}{\label}
\nc{\ba}{\begin{array}}
\nc{\ea}{\end{array}}
\nc{\bs}{\begin{split}}
\nc{\es}{\end{split}}

\newcommand{\R}{\mathbb{R}}
\newcommand{\C}{\mathbb{C}}
\newcommand{\Z}{\mathbb{Z}}

\newcommand{\frach}{\mathfrak{h}}

\newcommand{\cD}{\mathcal{D}}
\newcommand{\cC}{\mathcal{C}}
\newcommand{\cE}{\mathcal{E}}
\newcommand{\cF}{\mathcal{F}}
\newcommand{\F}{\mathcal{F}} 

\newcommand{\cH}{\mathcal{H}}

\newcommand{\cL}{\mathcal{L}}
 
\newcommand{\cP}{\mathcal{P}}         

\newcommand{\cS}{\mathcal{S}}

\nc{\e}{\epsilon}
\nc{\al}{\alpha}
\nc{\del}{\delta}
\nc{\h}{\delta}
\nc{\G}{\eta} 
\nc{\et}{\eta} 
\nc{\Gam}{\eta}  
\nc{\g}{\gamma}
\nc{\gam}{\gamma}
\nc{\ka}{\kappa}
\nc{\lam}{\lambda}
\nc{\Lam}{\Lambda}
\nc{\Om}{\Omega}
\nc{\om}{\omega}

\nc{\ta}{\tau}
\nc{\w}{\omega}
\nc{\io}{\iota}
\nc{\z}{\zeta}
\nc{\s}{\alpha}
\nc{\Si}{\Sigma}
\nc{\vphi}{\varphi}

\nc{\bP}{\bar{P}}
\nc{\bQ}{\bar{Q}}

\def\qf{\varphi}

\nc{\ran}{\rangle}
\nc{\lan}{\langle}

\newcommand{\ra}{\rightarrow}

\newcommand{\ls}{\lesssim}
\newcommand{\gs}{\gtrsim}
\newcommand{\one}{\mathbf{1}}

\renewcommand{\Re}{\mathrm{Re}} 
\newcommand{\Tr}{\mathrm{Tr}}
\newcommand{\tr}{\mathrm{Tr}}
\newcommand{\diag}{\mathrm{diag}}

\nc{\bfone}{{\bf 1}}
\nc{\1}{{\bf 1}}

\newcommand{\p}{\partial}

\newcommand{\n}{\nabla}

\newcommand{\curl}{\operatorname{curl}}
\newcommand{\CURL}{\operatorname{curl}}
\newcommand{\divv}{\operatorname{div}}

\renewcommand{\div}{\operatorname{div}}
\newcommand{\grad}{\operatorname{grad}}
\newcommand{\hess}{\operatorname{Hess}}

\newcommand{\na}{\nabla_a}

\newcommand{\COVGRAD}[1]{\nabla_{\!\!#1}}

\newcommand{\rbrac}[1]{\left(#1\right)} 

\def\eqn{\begin{align}}
\def\eeqn{\end{align}}

\newcommand{\DETAILS}[1]{}

\newcommand{\frachb}{\mathfrak{h}_{b}}
\newcommand{\frachbvec}{\vec{\mathfrak{h}}_{ b}}
\newcommand{\frachvec}{\vec{\mathfrak{h}}}

\newcommand{\lat}{\mathcal{L}} 
\newcommand{\Omd}{\Om}
\newcommand{\den}{\text{den}}
\newcommand{\sech}{\text{sech}}

\begin{document}


\title[Bogolubov-de Gennes Equations February 7, 2019]{On the Bogolubov-de Gennes Equations of Superconductivity}

\date{February 7, 2019} 
{\center\author{Ilias (Li) Chenn and I. M. Sigal}} 
\maketitle

\begin{abstract} We consider the Bogolubov-de Gennes equations giving an equivalent formulation of the BCS theory of  superconductivity.  We are interested in the case when the magnetic field is present. We (a) discuss their general features, (b) isolate  key physical classes of solutions (normal, vortex and vortex lattice states) and (c) prove existence of the normal, vortex and vortex lattice states and stability/instability of the normal states for large/small  temperature or/and magnetic fields. 
\end{abstract}


\section{Introduction} 

The Bogolubov-de Gennes equations describe the remarkable quantum phenomenon of  superconductivity.\footnote{For some physics background, see books \cite{dG, MartRoth} and the review papers \cite{Cyr, Legg}.} They present an equivalent formulation of the BCS theory and are among the latest additions to the family of important effective equations of mathematical physics. Together with the Hartree-Fock (-Bogolubov), Ginzburg-Landau and Landau-Lifshitz  equations, they are the quantum members of this illustrious family consisting of such luminaries as 
the heat, 
 Euler, Navier-Stokes and Boltzmann equations. 

 There are still many fundamental questions about these equations which are completely open, namely 
\begin{itemize}
\item 
Derivation;

\item Well-posedness;

\item Existence and stability of stationary magnetic solutions. 
\end{itemize}
 
By the magnetic solutions we mean (physically interesting) solutions with non-zero magnetic fields. In this paper we address 
the third problem.  The well-posedness (or existence) theory 
  will be addressed elsewhere (cf. \cite{BBCFS}). 

The key special solutions of Bogoliubov-de Gennes (BdG) equations are  normal, superconducting and mixed or intermediate states. The latter appear only for  non-vanishing magnetic fields. For type II superconductors, they consist of the vortices and (magnetic) vortex lattices. In this paper, we  prove the existence of the normal states for non-vanishing magnetic fields and of the vortex lattices and investigate the stability of the former. 

There is a considerable physics literature devoted to the BdG  equations, but, despite the role played by magnetic phenomena in superconductivity, it deals mainly with the zero magnetic field case, with only few disjoint remarks about the case when the magnetic fields are present, the main subject of this work.\footnote{The Ginzburg-Landau equations give a good account of magnetic phenomena in superconductors but only for temperatures sufficiently close to the critical one.} 

As for rigorous work, it also deals exclusively with the case of zero magnetic field.  The general (variational) set-up for the BdG equations is given in  \cite{BLS}. We use, like all subsequent papers, this set-up. 
The next seminal works on the subject are  \cite{HHSS}, where the authors prove the existence of superconducting states (the existence of the normal states under the assumptions of \cite{HHSS} is trivial), to which our  work is closest, and \cite{FHSS},  deriving the (macroscopic) Ginzburg-Landau equations. 
 For an excellent, recent review of 
 the subject, 
 with extensive references and discussion see \cite{HaiSei}. 
 
In the rest of this section we introduce the BdG  equations, describe their properties and the main issues and present the main results of this paper. In the remaining sections 
 we prove the these results, 
 with technical derivations delegated to appendices. 
  In the last appendix, following \cite{BBCFS}, we discuss a formal, but  natural, derivation of the BdG  equations.

\subsection{Bogolubov-de Gennes equations}\label{sec:bdg}

In the Bogolubov-de Gennes approach states of superconductors are described by the pair of operators $\g$ and $\s$, acting on the  one-particle state space, $\frach$,  and satisfying (after peeling off the spin variables)

\begin{equation}  \label{gam-al-prop} 
	0\le \gamma=\g^* \le 1,\  \quad 	 \s^*= \overline\s\ \quad  
	  {\rm\ and\ }\  \quad \alpha \alpha^* \leq \gamma(1-\gamma)
\end{equation}
where  $\overline\gamma :=\cC \g \cC$, with $\cC$, the operation of complex conjugation and $\al$ should thought of as acting from $\cC \frach$ to $\frach$.  $\g$ is a one-particle density operator, or diagonal correlation, and $\s$ is a two-particle coherence operator, or off-diagonal correlation. $\g(x, x)$ is interpreted as the one-particle density. 

  The one-particle space $\frach$ is determined by the many-body quantum problem. 
For zero density (or `finite') systems it is  $L^2(\R^d)$ and for positive density ones,  $L^2(\Om)$, where $\Om$ is a fundamental cell in a lattice $\lat\subset \R^d$, with magnetic filed dependent (twisted) boundary conditions (see Subsection \ref{sec:1pspaces} for more details). We take $\frach$ to be either $L^2(\R^d)$ or $L^2(\Om)$ and  understand $\int$ without specifying the domain of integration as taken over either $\R^d$ or $\Om$, depending on the choice of the one particle space.

The Bogoliubov-de Gennes equations form a system of self-consistent equations for  $\g$ and $\s$. 
It is convenient to organize the operators  $\gamma$ and $\s$ 
 into the self-adjoint operator-matrix 
\begin{align} \label{Gam}
	\G:= \left( \begin{array}{cc} \g & \s \\  \s^* & \one-\bar\g \end{array} \right).
\end{align}
The definition of $\g$ and $\al$ in terms of the many-body theory (see \eqref{omq-mom} of Appendix \ref{sec:qf-reduct}) implies that 
\begin{align} \label{Gam-prop} 
	0\le \G=\G^* \le 1\ {\rm\ and\ }\   J^* \G J=\one-\bar\G,\ J:=\left( \begin{array}{cc} 0 & \one \\  - \one & 0 \end{array} \right). \end{align} 
These relations imply relations \eqref{gam-al-prop}. 

Since  the BdG equations describe the phenomenon of superconductivity, they are naturally coupled to the electromagnetic field. We describe the latter by the vector  potential $a$ in the gauge in which the electrostatic potential is zero. 
Then the time-dependent BdG equations state (see e.g. \cite{dG, Cyr, MartRoth}) 

\begin{align}\label{BdG-eq-t}
& i \partial_t  \G =[\Lambda(\G, a), \G],\\ 
\label{hgam-a}
&\Lambda(\G, a) =  \big(\begin{smallmatrix} h_{\g a} & v^\sharp \s\\ v^\sharp \bar{ \s} & -\overline{h_{\g a}}\end{smallmatrix}\big),\ \quad  h_{\g a}  =-\Delta_a+ 
  v^* \gamma  - v^\sharp\, \gamma \,, 
\end{align} where $v(x, y)$ is the pair potential, the operator $v^\sharp$ 
  is defined through the integral kernels as $(v^\sharp\, \al) \, (x;y):= v(x, y)\al (x;y)$ and 
$(v^* \gamma)(x)= (v* \rho_\gamma)(x) :=\int v(x, y) \rho_\gamma(y) d y,$
with $\rho_\gamma(x):= \gamma(x, x)$.  
 $v^* \gamma$ and $v^\sharp\, \gamma$ are the direct and exchange self-interaction potentials, and  $\Delta_a:=|\n_a|^2, \n_a:=\n- i a$. 
 Eq \eqref{BdG-eq-t} is coupled to the  Maxwell equation (Amp\`{e}re's law)
\begin{align}\label{Amp-Maxw-eq}
 \partial_t^2 a   &= -\curl^* \curl a + j(\G, a),
\end{align}
where $j(\G, a)(x)\equiv j(\g, a)(x) := [-i \na,\g]_+(x, x)$ 
 is the superconducting current (in our gauge, the electric field is $E=-\p_t a$). (Above, $A(x, y)$ stands for the integral kernel of an operator $A$.) In what follows, we assume that 
\begin{align}\label{v-cond}
	v(x, y)=v(x - y),\ \text{ and }\ v(-x)=v(x).
\end{align}
We specify below additional assumptions on  $v(x,y)$ and on the operators $\g$ and $\al$  so that all the terms in  \eqref{BdG-eq-t} - \eqref{Amp-Maxw-eq} are well defined (at least weakly).

\medskip

 {\bf Connection with  the BCS theory. }  Eq \eqref{BdG-eq-t} can be reformulated as an equation on the Fock space involving an effective quadratic Hamiltonian  (see \cite{Cyr, dG, HaiSei} and \cite{BBCFS}, for the bosonic version). These are the effective BCS equations and the effective BCS Hamiltonian. 

\medskip

\begin{rem}\label{rem:space} 
{\em 1)  $v$ is not an electrostatic potential but an effective two-body one. To avoid a confusion we have chosen the gauge in which the electrostatic potential is zero. 

2) For $\s=0$, Eq \eqref{BdG-eq-t} 
 becomes 
 the  time-dependent von Neumann-Hartree-Fock equation for $\g$ (and $a$).

3) 
One can extend a formal derivation of \eqref{BdG-eq-t} given in Appendix \ref{sec:qf-reduct} to the coupled system \eqref{BdG-eq-t}-\eqref{Amp-Maxw-eq} by starting with the many-body quantum Hamiltonian \eqref{H} coupled to the quantized electro-magnetic filed. 

4) For a definition of $\rho_\g$ not relying on the integral kernels, see \eqref{den-def}.}
\end{rem}

\medskip
\subsection{Symmetries and conservation laws} 
\label{sec:bdg-prop} 
 The equations \eqref{BdG-eq-t} - \eqref{Amp-Maxw-eq} are invariant under the time-independent {\it gauge} transformations, 
\begin{equation}\label{gauge-transf}
    T^{\rm gauge}_\chi : (\g, \s, a) \mapsto (e^{i\chi }\g e^{-i\chi } , e^{i\chi }  \s e^{i\chi }, a + \nabla\chi), 
\end{equation}
  for any sufficiently regular function $\chi :  \R^d \to \R$,  and  
the {\it translation}, {\it rotation} and {\it reflection} transformations, 
\begin{align}\label{tr-transf}
   & T^{\rm trans}_h :  (\g, \s, a) \mapsto (u_h\g u_h^{-1}, u_h\s u_h^{-1}, u_h a),\\
\label{rot-transf}
  &  T^{\rm rot}_\rho : (\g, \s, a) \mapsto (u_\rho\g u_\rho^{-1} , u_\rho\s u_\rho^{-1}, \rho u_\rho a),\\
\label{refl-transf}
  &  T^{\rm refl} : (\g, \s, a) \mapsto (u^{\rm refl}\g (u^{\rm refl})^{-1} , u^{\rm refl}\s (u^{\rm refl})^{-1}, - u^{\rm refl} a),\end{align}
 for any $h \in \R^d$ and $\rho \in O(d)$. 
Here $u_h\equiv u_h^{\rm transl}$, $u_\rho\equiv u_\rho^{\rm rot}$ and $u^{\rm refl}$ are the standard translation, rotation and reflection transforms $u^{\rm transl}_h :  f(x) \mapsto f(x + h)$, $u^{\rm rot}_\rho : f(x)  \mapsto  f(\rho^{-1}x)$ and $u^{\rm refl} : f(x)  \mapsto  f(- x)$.  

In terms of $\G$, say the gauge transformation, $T^{\rm gauge}_\chi$, could be written as 
\begin{equation}\label{gauge-transf'} 
  T^{\rm gauge}_\chi : 	\G\ra U_\chi \G U_\chi^{-1},\ \text{ where }\ U_\chi=\left( \begin{array}{cc} e^{i\chi }  & 0 \\ 0  & e^{-i\chi }  \end{array} \right)
\end{equation}
(extended correspondingly to $(\G, a)$ by $T_\chi^{\rm gauge} (\G , a)=(T_\chi^{\rm gauge} \G , a+\nabla \chi)$). Notice the difference in action of this transformation on the diagonal and off-diagonal elements of $\G$.

The  invariance under the gauge transformations can be proven by using the relation
\begin{equation}\label{Lam-gauge-cov}
	\Lambda(T_\chi^{\rm gauge} (\G , a))=  T_\chi^{\rm gauge} (\Lambda(\G, a)),
\end{equation}	
 shown by using the operator calculus.

We also keep the notation $T^{\rm trans}_h$ for these operators restricted to $\G$'s.
 
\smallskip

Since we assumed that the external fields are zero, the equations are translationally invariant (when considered in $\R^d$). Because of the gauge invariance, it is natural to consider the simplest, gauge (magnetically) translationally invariant solutions, i.e. solutions invariant under the transformations   
\begin{align}\label{mag-transl}T_{b s}: = (T^{\rm gauge}_{\chi_{s}^b})^{-1}T^{\rm trans}_{ s}, 
\end{align}  
for any $s\in \R^2$, where $\chi_{b s} (x):\lat\times \R^2\ra \R$ is given in 
\begin{align}\label{chibs}\chi_{b s}(x):=  x\cdot a_b(s)+c_s,\end{align} where $a_b(x)$ is the vector potential with the constant magnetic field, $\curl a_b=b$  and $c_s$ are constants satisfying  
$ c_{s+t} - c_s - c_t - \frac{b}{2}  s \wedge t \in 2\pi\Z. $ 
   In the gauge we work below with, we have $\chi_{b s}(x):=  \frac b2 (s\wedge x)+c_s$. 
We have

\begin{lemma} \label{lem:mtProjRep}
The operators $T_{bs}$, defined in \eqref{mag-transl} and restricted to $\G$'s, 
 satisfy 
\begin{align} \label{TsTt}	& T_{bs}T_{bt} = \hat I_{b s t} T_{b(s+t)},\\
 \label{Ist}&\hat I_{b s t} u:= I_{b s t} u I_{b s t}^{-1},\ I_{b s t} := \left( \begin{array}{cc} e^{i \frac b2 (s \wedge t)} & 0 \\
0 & e^{-i \frac b2 (s \wedge t)}\end{array} \right)\end{align}
\end{lemma}
\begin{proof}
Let $U^{\rm trans}_{ s}:=\diag(u_s^{\rm trans}, u_s^{\rm trans})$  on $L^2_{\rm loc}(\R^2)\times L^2_{\rm loc}(\R^2)$ and define the transformations 
 \begin{align}\label{Ubs}U_{b s}:= (U^{\rm gauge}_{\chi_s})^{-1}U^{\rm trans}_{ s}.\end{align} Then $T_{b s} \G =U_{b s} \G U_{b s}^{-1}$ and $U_{b s} U_{b t}=I_{b s t} 
 U_{b s+t}=  U_{b s+t}I_{b s t} $ 
 where we used that $g_{s t}(x):= \chi_s^b(x)+ \chi_t^b(x+t) - \chi_{s+t}^b(x)=  \frac b2 (t\wedge s)$. 
  Hence, the result follows.
\end{proof}

\medskip

 \paragraph{\bf Particle-hole symmetry.}   The evolution under the equations \eqref{BdG-eq-t} - \eqref{Amp-Maxw-eq} preserves the relations in \eqref{Gam-prop}, i.e. if an initial condition has one of these properties, then so does the solution.  
This follows from the relation 
\begin{align} \label{Lam-ph-sym} J^*\Lam J= -  \overline{\Lam}.\end{align}
 The second relation in \eqref{Gam-prop} is called the particle-hole symmetry.

\medskip

 \paragraph{\bf Conservation laws.} 
  Eqs \eqref{BdG-eq-t} -- \eqref{Amp-Maxw-eq}  conserve the energy  
    \begin{align} \notag \cE(\G, a, \phi) :=& \Tr\big((-\Delta_a) \gamma \big) + \frac12 \Tr\big((v* \rho_\g) \gamma \big)
   -\frac12\Tr\big((v^\sharp \g) \gamma \big) + \frac12 \Tr\big( \s^* (v^\sharp \s) \big)\\
 \label{energy-t} & 
+ \frac12\int ( |\curl a|^2+ |\p_t a |^2). 
\end{align} 
(Recall that in our gauge, $-\dot a\equiv -\p_t a$ is the electric field.)  
 Indeed, taking the time derivative of \eqref{energy-t} and using that $d_a \Tr\big((-\Delta_a) \gamma \big)a'=-\Tr\big(j(\G, a) a' \big)$, we find 
 \begin{align} \label{dtE}	\partial_t \cE 
=& \tr(h_{a \g} \dot \g)  +\Tr\big( \dot\s^* (v^\sharp \s) +  \s^* (v^\sharp \dot\s)\big)\\
 \label{dtE'}& + \lan \curl^* \curl - j(\G, a), \dot a\ran + \lan \ddot{a} , \dot a \ran	
\end{align}
where the inner product and trace are in $\frach$. A simple computation shows that line \eqref{dtE} 
$= \frac12 \tr(\dot\G \Lambda(\G, a))$ and therefore by \eqref{BdG-eq-t} is $0$. Line \eqref{dtE'} vanishes by \eqref{Amp-Maxw-eq}. Hence  
 $\partial_t \cE =0$. \qquad \qquad $\Box$

Furthermore, the global gauge invariance implies the evolution conserves the number of particles, $N:=\Tr \g$.

\begin{rem}\label{rem:aut-fact}   
{\em   The trace of operator valued $2 \times 2$ matrices is defined as the $\mathfrak{h} \otimes \C^2$ trace, i.e. as the sum of traces of the diagonal entries individually in $\mathfrak{h}$. If a matrix-operator is trace class, then this trace coincides with the standard one.
 If $A \geq 0$, then the two forms of trace for the operator $A$ agree. More precisely, $A \geq 0$ is trace class if and only if its $\frach \otimes \C^2$ trace is finite. In particular, $S(\G) < \infty$ if and only if $g(\G)$ is trace class.
}
\end{rem}

 \medskip

\subsection{Stationary Bogoliubov-de Gennes equations}\label{sec:station-bdg}

We consider stationary solutions to \eqref{BdG-eq-t} of the form 
\begin{align}\label{stat-sol}
   \G_t :=  T^{\rm gauge}_\chi \G = U_\chi^{\rm gauge} \G  (U_\chi^{\rm gauge})^{-1},
\end{align}
 with $\G$ 
 independent of $t$,   $a$ independent of $t$ and $\phi=0$.  
 We have  \begin{proposition}\label{prop:stat-sol} 
Operator-family \eqref{stat-sol}, with $\G$ independent of $t$ and $\dot\chi\equiv -\mu$, constant, is a solution to \eqref{BdG-eq-t}, iff $\G$ solves the equation
\begin{equation}\label{BdG-eq-stat'}
[\Lam_{\G a}, \G ]=0, 
\end{equation}
where $\Lam_{\G a}\equiv \Lam_{\G a \mu}: = \Lambda(\G, a) - \mu S$, with, recall,  $S:=\left( \begin{array}{cc} 1 & 0 \\ 0 & - 1 \end{array} \right)$, and is given explicitly  
\begin{align} \label{Lam'} 
 \Lam_{\G a }:=\left( \begin{array}{cc} h_{\g a \mu} & v^\sharp \s \\ v^\sharp \s^* & - \bar h_{\g a \mu} \end{array} \right), \end{align}
with $h_{\g a \mu}:=h_{\g a}-\mu$.
\end{proposition}
\begin{proof}   
Plugging \eqref{stat-sol} into \eqref{BdG-eq-t},  and using \eqref{Lam-gauge-cov}, the relation $\partial_t (U_\chi \eta U_\chi^{-1}) = i\dot \chi [S, U_\chi \eta U_\chi^{-1}]= i\dot \chi U_\chi [S, \eta] U_\chi^{-1}$ and that $\chi$ is independent of $x$, we see that 
\begin{align}
	-\dot \chi U_\chi [S, \eta] U_\chi^{-1} 
	 =& [\Lambda(U_\chi \eta U_\chi^{-1}, a), U_\chi \eta U_\chi^{-1} ] \\
		=& [\Lambda(U_\chi \eta U_\chi^{-1}, a+\nabla \chi), U_\chi \eta U_\chi^{-1} ] \\
		=& U_\chi [\Lambda(\eta, a),\eta] U_\chi^{-1}.
\end{align}
 Since $\dot\chi\equiv -\mu$, 
 it follows then
$	[\Lambda(\eta, a)-\mu S,\eta]=0,$ which is exactly \eqref{BdG-eq-stat'} and therefore gives the statement of the proposition.
\end{proof}

 For any reasonable function $f$, solutions of the equation 
 \begin{align} \label{Gam-eq'} 
	\G=f(\frac{1}{T} \Lambda_{\G a}),
\end{align} 
solve \eqref{BdG-eq-stat'} and therefore give stationary solutions of \eqref{BdG-eq-t}. Under some conditions, the converse is also true. (The parameter $T>0$, the temperature, is introduced here for the future references.)

The physical function $f$ is selected by either a thermodynamic limit (Gibbs states) or by a contact with a reservoir (or imposing  the maximum entropy principle) 
 It is given by the Fermi-Dirac distribution
\begin{align} \label{FD-distr} 
	f_{\rm FD} (h) =(1+e^{ h})^{-1}.
\end{align} 
Inverting the function $f$, one can rewrite \eqref{Gam-eq'} as $\Lambda_{\G a}=T f^{-1}(\G)$. Let $f^{-1}=: g'$. Then the stationary Bogoliubov-de Gennes equations 
can be written (in the Coulomb gauge $\divv a=0$) as
\begin{align} \label{Gam-eq}
	&\Lam_{\G a} -T g'(\G) =0,\\ 
 \label{a-eq} 
     &\CURL^*\CURL a - j(\G, a)=0.
\end{align}
Here  $\Lam_{\G a}\equiv \Lam (\G, a)$ and $T\ge 0$ (temperature) and, as follows from the equations $g'=f^{-1}$ and \eqref{FD-distr},   
the function $g$ is given by 
\begin{equation} \label{g}
	g(\lam)= - (\lam\ln \lam  + (1-\lam) \ln (1-\lam)),
\end{equation}
so that
\begin{equation} \label{g'} 
	g'(\lam) =- \ln \frac{\lam}{ 1-\lam}. 
\end{equation} 
From now on, 
 {\it we assume that $g(\lam)$ is given in} \eqref{g}. 

\medskip

\begin{rem}\label{rem:BdG} 
{\em 1) One can express these equations in terms of eigenfunctions of the operator $\Lam_{\G a}$, which is the form appearing in physics literature (see \cite{BLS, BBCFS}).

 2)  For \eqref{Gam-eq'} to give $\G$ of the form \eqref{Gam}, the function  $f (h)$ should satisfy the conditions
\begin{equation} \label{g-cond}
	f(\bar h) = \overline{f(h)} \text{   and   } f(-h) =\one - f(h) .
\end{equation}
 For  $g(x)$ given in \eqref{g}, the function  $f(h):= (g')^{-1}(h)$ satisfies these conditions as can be checked from its explicit form  \eqref{FD-distr}.  
However, \eqref{g-cond} is more general than \eqref{FD-distr}. 
Indeed, the first condition in \eqref{g-cond} means merely that $f$ is a real function, 
  while the second condition in \eqref{g-cond} is satisfied by functions $f(h):= (1+e^{\tilde{g}(h)})^{-1}$, with  $\tilde{g}(h)$, any odd function. 
  \DETAILS{To check this we set $g(h) = f(h)^{-1}-1$. Then we see that it is equivalent to
\begin{align}
	g(h)g(-h) =1 .
\end{align}
It is easy to see now that $f(h):= (1+e^{\tilde{g}(h)})^{-1}$, with  $\tilde{g}(h)$ odd, satisfies \eqref{g-cond}.}

 More generally, one could require that  $g(\lam)$ satisfies the conditions \eqref{g-cond}, with $f(h):= (g')^{-1}(h)$, and 
\begin{align}\label{g-cond-2}   g'(1-x)=-g'(x).\end{align}

3) If we drop the direct and exchange self-interactions from $h_{\g a \mu}$, then $\Lam_{\G a }$ becomes independent of the diagonal part, $\g$, of $\G$ and the equation \eqref{Gam-eq'} implies that \eqref{Gam-eq} has always the solution}
\begin{equation}\label{Gam-sol-diag}
	\G_{T a }\ =f_{\rm FD} (\frac{1}{T} \Lam_{a }),\ \text{ where }\   \Lam_{a }:= \Lam_{\G a }\big|_{\G=0}.
\end{equation}
\end{rem}

 \subsection{Free energy} 
 With the operators $\g$ and $\al$ are defined on the space $\frach$, the static Bogoliubov-de Gennes equations arise as 
the Euler-Lagrange equations  for the free energy functional 
\begin{align}\label{FT-def} F_{T}(\G, a):=E(\G, a) -T S(\G)-\mu N(\G),\end{align}
where $S(\G) = \Tr g(\G)$ is the entropy, 
 $N(\G):=\tr \g$ is the number of particles, and $E(\G, a)$ is the energy functional \eqref{energy-t} for $\G$ and $a$  time-independent and is given (in the Coulomb gauge $\divv a=0$) by 
\begin{align} \notag 
	E(\G, a) &=\Tr\big((-\Delta_a) \gamma \big) +\frac12\Tr\big((v* \rho_\g) \gamma \big) -\frac12\Tr\big((v^\sharp \g) \gamma \big)\\
 \label{energy} & +\frac{1}{2} \Tr\big( \s^* (v^\sharp \s) \big) 
 +\int dx  |\curl a(x)|^2.
\end{align}

The energy functional $E(\G, a)$ originates as $E(\G, a):=\qf(H_a)$, where  $\qf$ is a quasi-free state in question (see Appendix \ref{sec:qf-reduct}) and  $H_a$  is the standard many-body given in \eqref{H}, coupled to the vector potential $a$.

\begin{rem}\label{rem:entr} 
{\em   Due to the symmetry \eqref{Gam-prop} of $\G$, we see that 
\begin{align}\label{S-sym}
	\Tr(\G \ln \G) = \Tr((1-\G)\ln(1-\G))
\end{align}
which, recalling \eqref{g}, implies that} 
\DETAILS{$\tr g(\G)=\tr s(\G)$, where 
\begin{align}
	s(\G):=-2 \G\ln\G \label{eqn:s-def}
\end{align}
 and  $g(\lam)$ is given in \eqref{g}, and therefore 
 \begin{align}\label{S-expr} S(\G) =\tr s(\G) = \Tr g(\G).\end{align}}
\begin{align}
\label{S-expr}	& S(\G) := \Tr(s(\G))= \Tr(g(\G)),\\
 \label{g-s-def}	&  g(\G): = -\G\ln\G - (1-\G)\ln(1-\G),\ s(\G ):= -2\G\ln\G.
\end{align}\end{rem}

 \subsection{Micro and macro  scales}

 We think of $v$ living on a microscopic scale, 
 normal and superconducting and mixed states as living on a macroscopic/ mesoscopic scale. 
  (the scale of the sample).  
Though in some related questions it is crucial to differentiate between different scales, this does not play a role in this paper and keep the notation simple we do not keep track of different scales in our notation. 

 \subsection{One-particle spaces}\label{sec:1pspaces}

 We introduce another - pre-thermodynamic limit - one - particle space.  To fix ideas, we assume in what follows, that $d=2$, which means effectively the cylinder geometry.

 Let  $\cL\subset \R^2$ is a Bravais lattice, fixed throughout the paper. 
By $\Omd$ we denote a fundamental cell of $\lat$. 
We introduce  the magnetic translation operator  
\begin{align}\label{ubs} u_{b s} := u^{\rm gauge}_{-\chi_{b s}} u^{\rm trans}_{ s},\end{align} where $\chi_{b s}$ is defined in \eqref{chibs} and, recall the operators $U^{\rm gauge}_{\chi}$ and $U^{\rm trans}_{ s}$ are defined as  $u^{\rm gauge}_\chi :    \phi(x) \mapsto e^{i\chi(x)}\phi(x)$ and $u^{\rm transl}_h :    \phi(x) \mapsto \phi(x + h)$ (cf. the definitions after \eqref{rot-transf}).
 Now, we define the periodic one-particle state space 
\begin{align}
	\frachb  := & \{ f \in L^2_{loc}(\R^2) : u_{b s} f = f \text{ for all } s \in \lat \}
\end{align}
 which is a Hilbert space with the scalar product defined, for an arbitrary fundamental cell $\Omd$ of $\lat$, as 
 \begin{align}\label{frachb-norm}\| f\|_{\frachb} :=\|f\|_{L^2(\Omd)}. 
\end{align} 
  In what follows, $\frach$ stands for either $\frachb$ or $L^2(\R^2)$.

\medskip

Similarly, we consider the Sobolev space of vector potentials:  let $\frachvec^{r}$ be either $H^r(\R^2; \R^2)$, if  the one-particle state space $\frach$ is $L^2(\R^2)$, or 
\begin{align}
\frachvec^{r}_{\rm per} := & \{ a \in H^r_{\rm loc}(\R^2; \R^2) : T_s a = a\ \forall s \in \lat, \div a=0, \int_\Om a = 0\},
\end{align}
for some fundamental cell $\Om$ of $\lat$, with the norm $\|a\|_{(r)}\equiv \|a\|_{H^r}\equiv \|a\|_{H^r(\Omd)}$, if  the one-particle state space is $\frachb$. (The conditions $\div a=0$ and $\int_\Om a  = 0$ make the the operator $\curl^*\curl$ strictly positive.)	Finally, we define the affine space
\begin{align}
\frachbvec^{r} := a_b + \frachvec^{r}.
\end{align}

Now, we define spaces of $\g$'s and $\al$'s used below. Let $I^r$ denote the space of bounded operators satisfying $\|A\|_{I^r}:=(\tr(A^*A)^{r/2})^{1/r}<\infty$ (a trace ideal or non-commutative $L^r-$space) and let $M_b := \sqrt{-\Delta_{a_b}}$. . 
We define 
Sobolev-type spaces for trace class operators by 
\begin{align} \label{Is1}	 
 	& I^{s,1} := \{ A : \frach  \rightarrow \frach  :  
	  \|A\|_{I^{s,1}} := \|M_b^{s} A M_b^{s}\|_{I^{1}} < \infty \},\\
\label{Is2} 
	&  I^{s,2} := \{ A : \frach  \rightarrow \frach : \, 
	 \|A\|_{ I^{s,2}} := \| A M_b^{s}\|_{I^{2}} 	< \infty \}.
\end{align}
Note that $I^{0,p} = I^p$. We will usually assume $\g\in I^{1,1}$ and $\al\in  I^{1,2}$. (Strictly speaking $\al$ acts from the dual space $\frach^*$ to $\frach$ (see \cite{BLS} for details), but for the sake of notational simplicity we will ignore this subtlety.) 
 We will use the  notation $\hat I^{s}$ for the space of $\G$'s with $\g\in I^{s,1}$ and $\al\in I^{s,2}$ and the norm
\begin{align} \label{eta-norm}	\|\G\|_{(s)} := \|\gamma \|_{I^{s,1}} + \|\al \|_{ I^{s,2}}.
\end{align} 
\DETAILS{	We denote by $I^{s,1, 2}$ 
	the Sobolev spaces of operators $\G$, with $\g\in I^{s,1}$ and $\al\in  I^{s,2}$, equipped with the norms
\begin{align}
	\|\G\|_{(s)} := \|\gamma \|_{I^{s,1}} + \|\alpha\|_{ I^{s,2}}
\end{align}}
Due to Lemma \ref{lem:al-gam-bnd}(2) below, $\g\in I^{1,1}\Rightarrow \al\in I^{1,2}$ and $\|\G\|_{(s)} \simeq \|\gamma \|_{I^{s,1}}$. 
 Furthermore, we define  
\DETAILS{\begin{align}
	\mathcal{D}^s := \mathcal{D} \cap I^{s,1, 2}, 
\end{align}
with $s=1$, where   {\bf(what is $\cL( \mathfrak{h} \times \mathfrak{h}^*)$?)} } 
\begin{align} \label{eqn:defDomainD}
	 \mathcal{D}^s_\nu 
	=\big\{\G\in \hat I^{s},\ 
	 \G \text{ satisfies \eqref{Gam-prop}},\  
	&\tr\g=\nu,\ S(\G)<\infty \big\}. 
\end{align}

\subsection{
Ground states of Bogoliubov-de Gennes equations}
\label{sec:bdg-spec-sol}

The static Bogolubov-de Gennes equations \eqref{Gam-eq}-\eqref{a-eq} have the following key classes of solutions providing candidates for the ground states:
\begin{enumerate}
\item Normal state:  $(\G, a)$, with $\al=0$. 

\item Superconducting state:  $(\G, a)$, with $\s \ne 0$ and $a = 0$.

\item Mixed state:  $(\G, a)$, with $\s \ne 0$ and $a \ne 0$.
\end{enumerate}

We discuss the above states in more detail. 

\medskip
\paragraph{\bf Superconducting states.} The existence of superconducting,  translationally invariant solutions is proven in \cite{HHSS} (see this paper and \cite{HaiSei} for the references to earlier results). 

\medskip

\paragraph{\bf Normal states.} For $b=0$, we can choose $a=0$. 
 In this case, the existence of  normal translationally invariant solutions was proven in \cite{Ha}.

For $b\ne 0$, the simplest normal states are  the magnetically translation (mt-) invariant ones,  i.e. ones satisfying
\begin{align}\label{mt-invar}T_{b s} (  \G, a) = (  \G, a),\end{align}  
for any $s\in \R^2$, where $T_{b s}$ is defined in \eqref{mag-transl}. 
 Here $\G$ acts on $L^2_{\rm loc}(\R^2)\times L^2_{\rm loc}(\R^2)$ which could be further specified as $\frach\times \frach$, i.e. as 
 either $L^2(\R^2)\times L^2(\R^2)$ or $\frachb\times \frachb$.  In the second case, the fact that $T_{b s}$ maps the set of bounded diagonal operator-matrices on $\frachb\times \frachb$ into itself  follows from the relation \eqref{TsTt} - \eqref{Ist} and the fact that $\hat I_{b s t}$ is an identity on diagonal operator-matrices.

 
\medskip

\paragraph{\bf Mixed states.}
 
The main candidate for a mixed state is a {\it vortex lattice}, i.e. a state, $(\G, a)$, satisfying $\al\neq 0$ and the equivariance condition 
\begin{align}\label{VL:equiv}T^{\rm trans}_{ s} (\G, a) =  T^{\rm gauge}_{\chi_{bs}} (\G, a)\  \text{ for every }\ s \in \lat. 
\end{align}
Recall that $\chi_{bs}$ are defined in \eqref{chibs}.

\begin{rem}\label{rem:aut-fact}   
{\em  The map \eqref{chibs} 
 satisfies the co-cycle condition 
\begin{align}\label{co-cycle}
    \chi_{s+t}(x) -\chi_s (x+t) - \chi_t(x) \in 2\pi \Z,\ \forall s, t\in \lat,
\end{align}
and is called the {\it summand of automorphy} (see \cite{Sig0} for a relevant discussion). (The map 
$e^{i \chi} : \lat\times\R^2 \to U(1)$, where 
 $\chi (x, s) \equiv \chi_s (x)$ is called the {\it factor of automorphy}.)
In fact, for $b$ satisfying, every map $\chi_s : \lat\times\R^2 \to \R$ satisfying \eqref{co-cycle} is equivalent to  one of \eqref{chibs}.} 
\end{rem}

\medskip

\paragraph{\bf Magnetic flux quantization.} Denote by $\Omd$ a fundamental cell of $\lat$. One has the following result
\begin{align}\label{mf-quant} \text{ For a vortex lattices: }\ \frac{1}{2\pi} \int_{\Omd} \curl a = c_1(\chi)\in \Z. 
\end{align} 
Here 
 $c_1(\chi)$ is the first Chern number associated to  the summand of automorphy $\chi : \lat\times\R^2 \to \R$: 
\begin{equation}\label{gs-equiv}
    c_1(\chi) =\frac{1}{2\pi} (\chi_{\nu_2}(x+\nu_1) - \chi_{\nu_2}(x) - \chi_{\nu_1}(x+\nu_2) + \chi_{\nu_1}(x)), 
\end{equation}
  for any basis $\{\nu_1, \nu_2\}$ in $\cL$. 
  (By the cocycle condition \eqref{co-cycle},  this quantity is independent of $x$ and of the choice of the basis $\{\nu_1, \nu_2\}$ and is an integer, for more discussion see \cite{Sig0}.) 

\subsection{Results} \label{sec:results}   
We formulate a technical definition used below. Let the reflection operator $\tau^{\rm refl}$ be given by conjugation by the reflections, $u^{\rm refl}: f(x)\ra f(-x)$. We say that a state $(\G, a)$ is {\it even (or reflection symmetric)} iff 
 \begin{align} \label{even} \tau^{\rm refl} \g =\g,\  \tau^{\rm refl} \al =\al\ \text{ and }\ u^{\rm refl} a = - a.
\end{align} 
The reflections symmetry of the BdG equations implies that if an initial condition is even then so is the solution every moment of time.


We begin with a basic result on the differentiability of the free energy functional   $F_{T}$ along the perturbations (`tangent vectors') at $(\G, a)$ from the following class 
\begin{align}\label{pert-class}& \cP_\G:=\{(\G', a') \in \hat I^1 \times \frachvec^1:  \text{ \eqref{Gam'-cond} holds } \},\\
\label{Gam'-cond}  
&J^* \G' J=-\bar\G',\   (\G')^2 \lesssim [\G (1-\G)]^2, \text{ and } \Tr(S_1\G') = 0,
\end{align}
where $S_1 = \text{diag}(1,0)$ and $J$ is defined in \eqref{Gam-prop}. 
 Conditions \eqref{Gam'-cond} 
  are designed to handle a delicate problem of non-differentiability of $s(\lam):=2 \lam \ln \lam$ at $\lam=0$, while allowing for sufficiently rich set to derive the BdG equations. 

\begin{theorem} \label{thm:BdG=EL}
(a) The free energy functional   $F_{T}$ is well defined on the space $ \mathcal{D}^1_\nu \times \frachbvec^1$.

(b) 
$F_T$ is continuously (G\^ateaux  or Fr\'echet) differentiable at $(\G, a) \in \mathcal{D}^1_\nu \times \frachbvec^1$, 
   with respect of perturbations $(\G', a')\in \cP_\G$. 

(c) If $0 < \G < 1$, strictly and  $(\G, a)$ is even in the sense of the definition \eqref{even}, 
 then critical points of $F_T$ satisfy 
 the  BdG equations for some $\mu$ (as a result of the constraint $\Tr \g = \nu$). 
 
(d)  Minimizers of $F_T$ over $\mathcal{D}^1_\nu \times \frachbvec^1$ are its critical points. 
\end{theorem}

This theorem is proven in  Appendix \ref{sec:energy}. For the translation invariant case, it is proven 
  in \cite{HHSS}. 
   In general case, but with $a=0$, 
   the fact that BdG is the Euler-Lagrange equation of BCS used in  \cite{FHSS}, but it seems with no proof provided. 

As a result of Theorem \ref{thm:BdG=EL}, we write the  Bogoliubov-de Gennes equations as 
\begin{align}\label{F'-eq} 
F'_{T }(\G, a)=0, \end{align}
where the map  $F'_{T }(\G, a)$ is defined by the r.h.s. of the static Bogoliubov-de Gennes equations \eqref{Gam-eq}-\eqref{a-eq} and can be thought of a gradient map of $F_T$. \DETAILS{Hence,   we can write the equations  \eqref{Gam-eq}--\eqref{a-eq} 
as 
\begin{align}\label{F'-eq} 
F'_{T }(\G, a)=0,\end{align}
 where $F'_{T }(\G, a)=( F'_{T \eta}(\G, a), F'_{T a}(\G, a))$ 
  is the formal G\^ateaux derivative (more precisely, the trace metric gradients) of  
  the free energy,  $ F_{T}(\G, a)$, in $\G$ and $a$.} 
\DETAILS{\begin{align}\label{F'T} 
	F'_{T }(\G, a):=\Lam_{\G a } -T g'(\G). 
\end{align} } 

\medskip

From now on, we consider only the cylindrical geometry, i.e. we assume the {\it dimension} is $2$.
  In the remaining results, we drop the exchange term $v^\sharp \gamma$. (We expect that these terms can be readily added back to the equation and will not drastically effect the proof below.) The existence of the mt-invariant normal states for $b\ne 0$ is stated in the following theorem proven in Section \ref{sec:norm-states-exist}: 

\begin{theorem}\label{thm:norm-state-exist} 
Drop the exchange term $v^\sharp \gamma$ and let $|\int v|$ be small. Then the BdG equations \eqref{Gam-eq} - \eqref{a-eq} 
 on the  space $I^{2,1}\times I^{2,2}\times \frachbvec^{2}$  have  
 a mt-invariant  solution, unique on the set of even (in the sense of the definition \eqref{even}),  mt-invariant states. 
 
Moreover, this solution is normal (i.e $\al=0$) and is of the form $(\G=\G_{T, b},\  a=a_b)$, where   
(cf. \eqref{Gam-sol-diag})
\begin{align}\label{Gam-Tb} 
		\G_{T b }:= \left( \begin{array}{cc} \g_{T b } & 0 \\ 
		0 & \one-\bar\g_{T b } 	\end{array} \right), 
\end{align}
with $\g_{T b }$  solving the equation 
\begin{align}\label{gam-Tb}\gamma= f_{\rm FD}(\frac{1}{T} h_{\g, a_b }),\end{align}
where $f_{\rm FD}=(g')^{-1}$ and is given in \eqref{FD-distr}, and $a_b(x)$ is the magnetic potential with the constant magnetic field $b$ ($\curl a_b=b$). 
\end{theorem}
The fact that $\G_{T b }$ in \eqref{Gam-Tb} is diagonal should not come as a surprise as $a_b$ has a constant magnetic field $b$ throughout the sample and it corresponds to a normal state. It can be also seen from the following elementary statement
\begin{proposition}\label{prop:mt-inv-al0}
  If $\G$ is mt-invariant, then $\alpha = 0$.
\end{proposition}
\begin{proof}
By \eqref{TsTt} - \eqref{Ist}, the mt-invariance, \eqref{mt-invar}, implies that $\alpha 
= e^{-ib s \cdot a_b(t)}\alpha$ for all $s,t \in \R^2$, which yields that $\alpha = 0$.
\end{proof}

We address the question of the energetic stability of the mt-invariant states. To this end we define, in the standard way, the Hessian of the free energy in $\G$ as
\begin{align}\label{Hess-eta}F_T''(\G_{*}, a_*):=d_\G \grad_\G F_T(\G_{*},a_*),\end{align}
where $d_\G $ is the G\^ateaux derivative w.r.to $\G$ and $\grad_\G$ is the gradient w.r.to $\G$, defined by the equation $\tr( \grad_\G F(\G_{*},a_*), \G')= d_\G F(\G_{*},a_*) \G'$.

We consider $F_T''(\G_{*}, a_*)$ at $\G=\G_{T b}$ along physically relevant perturbations 
 of the form $\G' = \phi(\alpha)$, where $ \phi(\alpha)$ is 
the off-diagonal operator-matrix, defined by 
\begin{align}\label{Phi-al}	 	\phi(\al) :=\left( \begin{array}{cc}
											0 & \al \\
											\al^* & 0
										\end{array} \right) ,
\end{align}
 and $\al$ is a Hilbert-Schmidt operator on $\frach$ 
(which allows us to use the Hilbert space techniques.)
 
Moreover, we require that perturbations  $\G'= \phi(\alpha)$ satisfy the condition \eqref{Gam'-cond} at $\G=\G_{T b} $ which is equivalent to requiring that 
\begin{align}\label{al-cond}
	\alpha \alpha^* \lesssim [\gamma_{Tb}(1-\gamma_{Tb})]^2.
\end{align}

Let $h^L$ and $h^R$ stand for the operators acting on other operators by multiplication  from the left by the operator $h$  and from right by the operator $\bar h$ , respectively, and recall $v^\sharp$ is defined after \eqref{BdG-eq-t}. 
We have the following
 \begin{proposition}\label{prop:FT''}
For off-diagonal perturbations $\G' = \phi(\alpha)$, $F_T''(\G_{Tb}, a_b) \phi(\alpha)= \phi(L_{Tb}\alpha)$, where the operator  $L_{Tb}$ 
  is given by 
\begin{align}\label{LTb}
	& L_{Tb} := K_{Tb} + v^\sharp,\\
\label{KTb}	& K_{Tb} := \frac{h_{T b}^L+ h_{T b}^R}{\tanh(h_{T b}^L/T)+\tanh( h_{T b}^R/T)}, 
\end{align} 

  where $h_{T b}:=h_{\gamma_{Tb}, a_b}$,  on the space  of Hilbert-Schmidt operators on $\frach$. 
   \end{proposition} 
 Let $\langle \alpha,  \alpha' \rangle:= \Tr(\al^*\al')$. We say that $\G'$ is an {\it off-diagonal perturbation} iff  $\G' = \phi(\alpha)$ with $\al$ a Hilbert-Schmidt operator satisfying  \eqref{al-cond}. 
The next result 
  generalizes that of \cite{HHSS} for $a=0$: 

   \begin{proposition}\label{prop:FT-expan-order2}  
For off-diagonal perturbations $\G' = \phi(\alpha)$, 
 we have
\begin{align}\label{cF-expan2ord}
	F_T(\G_{Tb}+ \e\G', a_b) =& F_T (\G_{Tb},a_b) + \e^2\langle \alpha, L_{Tb} \alpha \rangle + O(\e^3).
\end{align}
 \end{proposition}

 Definition  \eqref{energy} of $E(\G, a)$ implies that $E(\G_{Tb}+ \e\G', a_b) = E(\G_{Tb}, a_b) + \e^2\Tr(\bar{\alpha} v^\sharp \alpha)$ for $\G'=\phi(\alpha)$. 
  This together with Corollary \ref{cor:S'-S''} and Eq \eqref{S-expan} of Appendix \ref{sec:entropy} on the entropy 
   yields Propositions \ref{prop:FT''} and \ref{prop:FT-expan-order2}, respectively.
\DETAILS{Now, proceed to the proof of Proposition \ref{prop:FT-expan-order2}. Using the definition \eqref{FT-def}, with \eqref{energy}, and the 
expansion  \eqref{S-expan}, we find, for perturbations $\G' = \phi(\alpha)$ obeying the condition \eqref{Gam'-cond}, 
\begin{align}
	F_T(\G_{Tb}+ \e\G', a_b) =& F_T(\G_{Tb}, a_b) + \e^2\Tr(\bar{\alpha} v^\sharp \alpha) +  \e^2 S''(\G',\G') + O(\epsilon^3). 
\end{align}}
(The absence of the linear term is due to the fact that $(\G_{Tb},0)$ is a critical point of the energy function.) 

The next two propositions are proven in Section \ref{sec:normal-stab/instab}.
  \begin{proposition}\label{prop:T-b-large}
 $L_{Tb}\ge  \frac12 T-\|v\|_\infty$  and consequently,  for $T$ sufficiently large, $L_{Tb}>0$ and the  normal state $(\G_{T, b},  a_b)$ is energetically stable under off-diagonal perturbations $\G' = \phi(\alpha)$. 
   \end{proposition} 

On the other hand,  for 
 $T$ and $b$ sufficiently small, we have
   \begin{proposition}\label{thm:T-b-small}
Suppose that $v < 0$ almost everywhere. 
 Then, for  $T$ and $b$ sufficiently small,  the operator $L_{Tb}$ acting on the set of Hilbert-Schmidt operator on $\frach$ has a negative eigenvalue and consequently 
  the  normal state $(\G_{T, b},  a_b)$ is energetically unstable under general $\al-$perturbations. 
   \end{proposition}
 Note that $L_{Tb}$ is discontinuous at $T=0$. 

Let $T_c(b)/T_c'(b)$ be the largest/smallest temperature s.t.  the normal solution is energetically unstable/stable under off-diagonal perturbations $\G' = \phi(\alpha)$, for $(T<T_c(b))/(T>T_c(b)')$. Clearly, $\infty \ge T_c'(b)\ge T_c(b)\ge 0$. 
Propositions \ref{prop:T-b-large} and  \ref{thm:T-b-small} imply 
\begin{corollary}\label{cor:Tcb-posit} 
Suppose that $v \le 0,\ v\not\equiv 0$. Then $T_c(b)>0$ for $b $ sufficiently small and  $T_c'(b) =0$ for $b $ sufficiently large. \end{corollary}
We conjecture that $T_c(b)= T_c'(b) .$

The next corollary provides a convenient criterion for the determination of $T_c(b)$ and $T_c(b)' $.
 \begin{corollary}\label{cor:Tc-d F_T'ev0}
At $T=T_c(b)$ and $T=T_c'(b) $, zero is the lowest eigenvalue of the operator  $L_{Tb}$. 
   \end{corollary}
 A proof of energetically stability under general perturbations for either $T$ or $b$ sufficiently large is more subtle. 
For it, 
 one has to use the full linearized operator. 
Our computations 
 suggest that $0$ is the lowest eigenvalue of $L_{Tb}$ iff $0$ is the lowest eigenvalue of $\hess F_T(\G_{Tb},a_b)$ and, consequently, $T_c(b)$ and $T_c'(b) $ apply also to the general perturbations. 

The statement $T_c=T_c(0)=T_c' (0)>0$ for $a=0$ and therefore $b=0$ and for a large class of potentials is proven, 
 by  the variational techniques, in \cite{HHSS}. 
 
In conclusion of this paragraph, we mention that a simple computation shows
  \begin{proposition}\label{prop:LTb-comm-MT}
The operator  $L_{Tb}$ commutes with the magnetic translations. The same is true for the $\G-$Hessian $F_T''(\G_{Tb},a_b)$ (see \eqref{Hess-eta}).  \end{proposition}

\paragraph{\bf Existence of vortex lattice solutions.} With the spaces defined above, 
 we have the following result on the existence of vortex lattices proven in Section \ref{sec:vort-latt-exist}: 
 
\begin{theorem} \label{thm:BdGExistence}
 Drop the self-interaction terms $2v* \rho_\gamma$ and $-v^\sharp \gamma$ in \eqref{hgam-a} and 
 assume that $T \ge 0$ and $\|v\|_\infty$ is small. We specify  $\frach=\frachb$. Then 

(i)  for a every Chern number $c_1$, there exists a (generalized)  solution $(\G, a) \in \mathcal{D}^1_\nu \times \frachbvec^1$ of the BdG equations \eqref{Gam-eq}--\eqref{a-eq} (in particular, it satisfies $T^{\rm trans}_{ s} (\G, a) =\hat T^{\rm gauge}_{\chi_{bs}} (\G, a)$), which  minimizes the free energy $F_T$ 
  for the given $c_1$;


(ii) if, in addition, $v\le 0, v\not\equiv 0$, then, for $T$ and $b$ are sufficiently small, 
$(\G,a)$ has $\alpha \not= 0$, i.e. this solution is a vortex lattice. More generally, the latter holds if  the operator $L_{Tb}$, given in \eqref{LTb}  and defined on $ I^{2, 2}$,
has a negative eigenvalue.
\end{theorem}


\begin{rem}\label{rem:LTb}  {\em 1)  The question of when $L_{Tb}$ has negative spectrum for a larger range of $T$'s is delicate one. For $T$ close to $T_c$, this depends, besides of the parameters $T$ and $b$, also on whether the superconductor is of Type I or II. 

2) 
Since the components of magnetic translations \eqref{mag-transl} do not commute, the fiber decomposition of $L_{Tb}$ is somewhat subtle (see \cite{AHS}). 

3)  One would like to remove the condition in Theorems \ref{thm:BdG=EL} and \ref{thm:norm-state-exist} and in the proof of Theorem \ref{thm:BdGExistence} that $(\G, a)$ is even in the sense of the definition \eqref{even}.

4) One would like to have a geometrical interpretation of the middle relation in \eqref{Gam'-cond}.
  
 5) The self-interaction terms $ v^* \gamma$ and $v^\sharp \gamma$ in \eqref{hgam-a} are inessential for physics and analysis and  are dropped for simplicity. 
 We expect that these terms can be readily added back to the equation and will not drastically effect our proof. 
 
 6) We derive the existence of the solutions from the existence of the minimizers of free energy \eqref{FT-def} and Theorem \ref{thm:BdG=EL}. 

7) One can likewise perform minimization among diagonal $\G$'s. This way, one obtains a variational proof of Theorem \ref{thm:norm-state-exist} on the existence on existence of normal states. 
  }
\end{rem}


 The paper is organized as follows. In Sections \ref{sec:norm-states-exist} and \ref{sec:vort-latt-exist}, we prove Theorems \ref{thm:norm-state-exist} and \ref{thm:BdGExistence}, on existence of the normal and vortex lattice solutions, respectively. These are our principal results. In Section \ref{sec:normal-stab/instab}, we prove Propositions  \ref{prop:T-b-large} and \ref{thm:T-b-small} on the stability/instability of the normal solutions.  
 In Appendix \ref{sec:entropy} we study the entropy functional. Results of this appendix imply the proofs of Propositions \ref{prop:FT''} and \ref{prop:FT-expan-order2} and are used in the proof of technical Theorem \ref{thm:BdG=EL} given in Appendix \ref{sec:energy}. 
  In Appendix \ref{sec:mt-den} we give another proof of the key statement of Section \ref{sec:norm-states-exist} (Proposition \ref{prop:gam-mt-invar-J0}) and  in Appendix \ref{sec:xi-fp}, we prove a technical result from that section. Finally,  in Appendix \ref{sec:den-est} we prove some bounds on functions relative to magnetic Laplacian and bounds on densities (both elementary and probably well-known)  and in Appendix \ref{sec:qf-reduct}, we discuss the derivation of the BdG equations from the quantum many-body problem.


\section{The normal states: 
 Proof of Theorem \ref{thm:norm-state-exist}}\label{sec:norm-states-exist} 
 
For normal states, i.e. for $\alpha = 0$, the BdG equations reduces to the following equations for $\g$ and $a$ (the coupled Hartree-Fock and Amp\'ere equations),
\begin{align}
	& \gamma = g^\sharp(\frac{1}{T} h_{\g, a }),  \label{eqn:NormalBdGPart1} \\
	& \curl^* \curl a = j(\g,a) \label{eqn:NormalBdGPart2}
\end{align}
where, recall, $j(\g, a) :=\den([-i \na,\g]_+)$.
First, we show first that the second equation is automatically satisfied for $a=a_b$ and $\gamma$ a magnetically translation invariant operator, which is even in the sense of \eqref{even}. 


 
\begin{lemma}\label{lem:grad-den} With  the magnetic translations $u_{b s}$ defined in \eqref{ubs} (cf. \eqref{Ubs}) and  $\tau_{b s}(\g)=u_{b s}\g {u_{b s}}^{-1}$, we have (cf. \eqref{TsTt})
\begin{align} \label{usuh'}	& u_{b s} u_{b h} = e^{-i b h\wedge s}  u_{b h}u_{b s},\\
\label{TsTt'}	  &\tau_{b s} \tau_{b h}=\tau_{b h}\tau_{b s},\\   
&\label{frachb-invar} \text{If $\g$ maps $\frachb$ into itself, then so does } \tau_{b h}\g.\end{align} 
\end{lemma}
\begin{proof} 

 \eqref{usuh'} follows from $u_{b s} u_{b h} = e^{-i\frac b2 h\wedge s} e^{-i\chi_{b s}}  u_{b h}\tau_s^{\rm transl}= e^{-i \frac b2 h\wedge s} e^{i\frac b2 s\wedge h}u_{b h}u_{b s}$. 
 To prove \eqref{TsTt'}, we 
 use that $\tau_s^{\rm transl} e^{-i\chi_{b h}}=e^{-i\frac b2 h\wedge s} e^{-i\chi_{b s}} \tau_h^{\rm transl}$, where, recall, $\chi_{b s}$ is defined in \eqref{chibs}, 
and \eqref{usuh'} which yields $(u_{b s} u_{b h})^{-1}=e^{i b h\wedge s} (u_{b h}u_{b s})^{-1}$ and therefore, due to  $\tau_{b h}\g:=u_{b h}\g u_{b h}^{-1}$, \eqref{TsTt'} follows.

Finally,   $\g: \frachb \ra \frachb \Longleftrightarrow \tau_{b s}(\g)=\g,\  \forall s\in \lat$. On the other hand, \eqref{TsTt'} implies $\tau_{b s} \tau_{b h}(\g)=\tau_{b h}\tau_{b s}(\g)=\tau_{b h}(\g),\  \forall  s\in \lat, h\in \R^2$, so \eqref{frachb-invar} follows. 
\end{proof}
Note that \eqref{usuh'} shows that $u_{b h}$ does not leave $\frachb$ invariant. 

The generators of the magnetic translations, $u_{bs}$ defined in \eqref{ubs}, and their properties are described in the following
\begin{lemma}[\cite{AHS}]\label{lem:gen-mt}
Let $p_{b} = -i\nabla_{a_b}$ and $\pi_{b} = -\bar p_{a_b}$, with the components $p_{b i}$ and $\pi_{b i}$. Then 
\begin{enumerate}
	\item $-i\partial_{s_j} \mid_{t=0} u_{b st} =  \pi_{b j}$; 
	\item $[\pi_{b i}, p_{b j}] = 0$ and therefore $[u_{bs}, p_{b j}] = 0$. 
\end{enumerate}
\end{lemma}

 For an operator $A$ on $\frachb$, we can define the integral $A'(x, y)$ by the relation
 \begin{align}\label{A-A'-rel}\lan g\otimes \bar f, A'\ran_{\frach\otimes \frach} 
=\lan g, A f\ran_{\frach},\ \quad 
\forall f, g\in \frach. \end{align}

In this section, it is convenient to use the notation den$(A)$ for the one-particle density $\rho_A$. 

For a trace-class operator $A$ on $\frach$, the density den$(A)\equiv \rho_A$ obeys the  relation  
\begin{align}\label{den-def}\int f \den (A)=\Tr( f A),\ \quad \forall f\in L^\infty.\end{align} 
which can be also used as a definition of den$(A)$. 

In the case of the space $\frachb$, we assume that all operators below are originally defined on $L^2_{\rm loc}(\R^2)$ (or on a local Sobolev space). This allows us to define compositions and commutators of operators some of which do not leave $\frachb$ invariant. 
Our key result here is the following
\begin{proposition}\label{prop:gam-mt-invar-J0} 
(i) If a trace-class operator $A$ on $\frach$ satisfies $\tau_{b h} A = A,\ \forall h\in \R^2$, then den$(A)$ is constant. (ii) If, in addition, $ \tau^{\rm refl} A =- A$ (with  the reflection operator $\tau^{\rm refl}$ defined before \eqref{even}),   then den$(\tilde\g)=0$.
\end{proposition}
\begin{proof} 
Recall that $\nabla_a:=\nabla - i a$. We begin with 
\begin{lemma}\label{lem:grad-den}
For any linear vector field $a$ and any integral operator $A$ on $\frach$,  $[\n_{a}, A]$  leaves $\frach$ invariant (though $\frachb$ is not invariant under  $\nabla_a$) , 
 \begin{align}\label{grad-den}\nabla \den(A)  = \den([\nabla_a, A])\end{align}
\end{lemma}
\begin{proof}
The statement is straightforward for $L^2(\R^2)$, so we consider $\frachb$. Since $a$ is linear, the invariance of $\frachb$ under  $[\n_{a}, A]$  is straightforward. 
We prove \eqref{grad-den}.

We have $\den ([\nabla_{a}, A]) =\den ([\nabla_{a_b}, A]) + i\den([a_b-a, A])$. Since $\nabla_{a_b}$ leaves $\frach$ invariant, we can use  the cyclicity of the trace to compute  
\begin{align}
	\int_\Om f \den ([\nabla_{a_b}, A]) &= \Tr_{\frachb}(f[\nabla_{a_b}, A])\\ 
		&= -\Tr_{\frachb}(\nabla f A)= -\int_\Om \nabla f \den(A). 
\end{align}
For any (linear) vector field $c$, the integral kernel of $[c, A]$ is $(c(x)-c(y))A(x, y)$ {\bf (see \eqref{A-A'-rel})}. Hence, we see that $\den([c, A])=0$ 
and therefore $\int_\Om f \den ([\nabla_{a}, A])= -\int_\Om \nabla f \den(A)$ for any $f\in L^\infty(\Om)$. Hence $\den ([\nabla_{a}, A])=\nabla  \den[A]$. 
\end{proof}

Since by Lemma \ref{lem:gen-mt}, $\tau_{b h}$ is generated by $A\ra i [\pi_{b}, A] $, the mt-invariance of $A$ implies that $[\pi_{b}, A]=0$. (Though $\pi_{b}$ does not leave $A$ invariant, $[\pi_{b}, A]$ does.) This and \eqref{grad-den} yield that $\nabla \den(A) =  \den([\pi_b, A])=0$ and therefore den $ A$ is constant.
\end{proof}

\begin{rem}\label{rem:another-pf-mt-inv-sol} {\em   
1) The long argument above proving \eqref{grad-den} establishes the intuitive fact that the integral kernel of the operator $[\nabla_a, A]$ acting on $\frachb$ is same as the integral kernel of this operator acting on $L^2(\Om)$, which is $(\nabla_{a x}+\overline{\nabla_{a y}})A'(x, y)=(\nabla_{ x}+\nabla_{ y})A'(x, y)$ and consequently $\den([\nabla_a, A])=(\nabla_{ x}+\nabla_{ y})A'(x, y)|_{x=y}=\n \den(A)$.

2) For another proof of the first statement of Proposition \ref{prop:gam-mt-invar-J0}(i) see Appendix \ref{sec:mt-den}.}
\end{rem}

If $\g$ is magnetically translationally invariant and even, then  $\tilde\g=-i \COVGRAD{a_b}\g$ is a magnetically translationally invariant and odd.  Applying this proposition to $\tilde\g=-i \COVGRAD{a_b}\g$, where $\g$ is a magnetically translationally invariant and even operator,  
  gives the equation  $ j(\g, a_b)=0$ and therefore, since $\curl^*\curl a_b=0$,
  \eqref{eqn:NormalBdGPart2} with $a=a_b$.

Now, we solve the first equation \eqref{eqn:NormalBdGPart1} for magnetic translation invariant $\g$'s. If we drop both the direct and exchange self-interactions from $h_{\g a \mu}$, then the latter equation becomes the definition of $\g_{T b }$: $\g_{T b }= f_{\rm FD}(\frac{1}{T} h_{ a_b })$, where, recall, $h_{ a_b }:=-\Delta_{a_b} - \mu$. Otherwise, we have to treat the equation $\g = f_{\rm FD}(\frac{1}{T} h_{\g a_b })$ as a fixed point problem. This problem simplifies considerably if we drop the exchange term, as in this case it reduces to a fixed point problem for a real number $\xi$: 
\begin{align}\label{xi-fp}
	\xi = v*\den(f_{\rm FD}((h_{ a_b } + \xi)/T)). 
\end{align}
Then $\xi$ is a real since $-\Delta_{a_b}$ is self-adjoint and is a multiple of the identity map by magnetic translation invariance due to Proposition \ref{prop:gam-mt-invar-J0}. 
Suppose that a real $\xi$ solves \eqref{xi-fp}, then define $\gamma := f_{\rm FD}((h_{ a_b } + \xi)/T)$. Since $f_{\rm FD} > 0$ and $h_{ a_b } + \xi$ is self-adjoint (for real $\xi$), we see that $\gamma \geq 0$. Then, we see that
\begin{align}
	\gamma :&= f_{\rm FD}((h_{ a_b } + \xi)/T) \label{eqn:gammaDefInTermsOfxi} \\
		&= f_{\rm FD}(h_{ a_b } + v*\den(f_{\rm FD}((h_{ a_b } + \xi)/T)))/T) \\
		&= f_{\rm FD}((h_{ a_b } + v*\den(\gamma))/T).
\end{align}
Hence $\gamma$ satisfies \eqref{eqn:NormalBdGPart1}. 
Conversely, if $\gamma$ solves the BdG equation \eqref{eqn:NormalBdGPart1}, then $\xi = v*\den(\gamma)$ satisfies the equation  
\begin{align}
	\xi = v*d_\gamma &= v*\den(f_{\rm FD}((h_{ a_b } + v*\den(\gamma))/T)) \\
		&= v*\den(f_{\rm FD}((h_{ a_b } + \xi)/T)).
\end{align}
So this $\xi$ is solution to \eqref{xi-fp}. We have therefore shown the following 

\begin{lemma} \label{lem:gbExistenceAndUniqueness}
There exists a solution $(\gamma_b, \alpha=0,a=a_b)$ with $\gamma_b \geq 0$ and $\gamma_b$ is a function of $-\Delta_{a_b}$ if and only if the fixed point problem \eqref{xi-fp} has a solution. Moreover, this solution is unique.
\end{lemma}

We show in Appendix \ref{sec:xi-fp} that the fixed point problem \eqref{xi-fp} has a unique solution if $|\int v |$ is small. Thus, we obtain an unique magnetic translation invariant solution. 
So we prove uniqueness among the class of $\gamma$'s which are functions of $-\Delta_{a_b}$.

Assume that $\gamma_1,\gamma_2$ are two solutions to \eqref{eqn:NormalBdGPart1}. Then we may form the corresponding $\xi_i = v*\den(\gamma_i)$ for $i=1,2$. Uniqueness of solution of equation \eqref{xi-fp} dictates that $\xi_1=\xi_2$. Therefore,
\begin{align}
	\gamma_1 =& f_{\rm FD}((h_{ a_b }+v*\den(\gamma_1))/T) \\
		=& f_{\rm FD}((h_{ a_b }+\xi_1)/T) \\
		=& f_{\rm FD}((h_{ a_b }+\xi_2)/T) \\
		=& f_{\rm FD}((h_{ a_b }+v*\den(\gamma_2))/T) = \gamma_2.
\end{align}

What remains to be done is to show that the solution is unique among solutions $(\gamma,a)$ such that $\gamma$ is magnetic translation invariant. It suffices to show that $a=a_b$, then equation \eqref{eqn:NormalBdGPart1} shows that $\gamma$ is a function of $-\Delta_{a_b}$ and we can conclude uniqueness by Lemma \ref{lem:gbExistenceAndUniqueness}. We decompose $a=a_b+a'$, where $a'$ is defined by this expression. Using equation \eqref{eqn:NormalBdGPart2}, we see that
\begin{align}
	\curl^* \curl a' = -a'(x)\gamma(0,0),
\end{align}
since, by Proposition \ref{prop:gam-mt-invar-J0}, $\gamma(x,x)$ is constant, the term $\Re(-i\nabla_{a_b}\gamma)(x,x)$ vanishes, and $\curl^* \curl a_b = 0$. Multiplying both sides by $a'$ and integrate, we see
\begin{align}
	\int |\curl a'|^2 + \gamma(0,0) \int |a'(x)|^2 = 0
\end{align}
Since $h_{ a_b }+\xi$ is bounded below and $g^\sharp$ is strictly positive and increase, we see that $\gamma = f_{\rm FD}((h_{ a_b }+\xi)/T) \geq c > 0$. Thus $\gamma(0,0) > 0$. It follows that $a' = 0$ and the proof is complete.

Finally, it is shown  in Appendix \ref{sec:xi-fp} that $\xi$ is smooth in $T$:
\begin{lemma} \label{lem:ExpandOfXi}
Assume that $\int v \leq 0$ and $|\int v|$ is small. Then $\xi$ is negative for all $T > 0$ and is bounded for all $T$ small. It is smooth and has the expansion, for $T$ small and with  $B = \frac{\hat v(0)}{4\pi}$,
\begin{align} \label{eqn:xiT-expansion}
	\xi(T) = \frac{B\mu}{1+B} + \frac{TB}{2} e^{\frac{-2\mu}{(1+B)T}} + O(TBe^{-\frac{4\mu}{(1+B)T}} + b^2). 
\end{align}
\end{lemma}


\section{Stability/instability of the  normal states for small $T$ and $b$: Proof of Propositions  \ref{prop:T-b-large} and \ref{thm:T-b-small}}  
 \label{sec:normal-stab/instab}

 \begin{proof}[Proof of Proposition \ref{prop:T-b-large}]
 Recall that $K_{Tb}=T f(h_x/T, h_y/T)$, where $f(u, v):=\frac{u+v}{\tanh(u)+\tanh(v)}$ and $h_z$ is the operator $h_{Tb}$, defined in Proposition \ref{prop:FT''}, acting on the variable $z$. By Lemma \ref{lem:f-low-bnd} below $f(u, v)\ge 1$. (A  weaker bound $f(u, v)\ge \frac14$ which suffices for us could be easily proved directly.) \DETAILS{Indeed, assume for simplicity that $x, y\ge 0$ and write 
\begin{align}\label{f-fn}
	f(x, y)= \frac{(x+ y)(1+e^{-2x})(1+e^{-2y})}{2(1-e^{-2(x+y)})}.
\end{align} 
Since $1+e^{-u}\ge 1$ and $1-e^{-u}\le u$, for $u\ge 0$, this gives $f(u, v)\ge \frac14$ and therefore} Hence 
\begin{align}\label{KTb-bnd}
	K_{Tb}\ge  T.
\end{align}  
Hence, since $v$ is independent of either $T$ or $b$, we have shown that $L_{Tb}\ge  T-\|v\|_\infty$  and consequently, Proposition \ref{prop:T-b-large}.  
 \end{proof} 
 
  
  \begin{proof}[Proof of Proposition \ref{thm:T-b-small}]

 We use the Birman-Schwinger principle (BSP) to show that $L_{Tb}$ has a negative eigenvalue. 
   Set $w^2 = -v \ge 0$ so 
  that $L_{Tb} = K_{Tb} - w^2$. 
  
  By the BSP,  $L_{Tb}$ has a negative eigenvalue $-E$ if and only if $G_{Tb}(E):= w(K_{Tb}+E)^{-1} w$  has the eigenvalue $1$ for some $E > 0$ (see e.g. \cite{GS}). 
By \eqref{KTb-bnd}, we have  $G_{Tb}(E)\ge 0, E\ge 0$.  Moreover, since $(K_{Tb}+E)^{-1}$ is continuous and monotonically decreasing in $E\ge 0$, so is $G_{Tb}(E)$. Hence, it suffices to show that $G_{Tb}:=G_{Tb}(0)$ satisfies $\|G_{Tb}\|>1$.
Hence, we estimate $\|G_{Tb}\|$ from below. 

 
Recall that $K_{Tb}=T f(h_x/T, h_y/T)$, where $f(u, v):=\frac{u+v}{\tanh(u)+\tanh(v)}$ and $h:=-\Delta_{a_b} -\mu$. 
Since the operator $h_{Tb}$, defined in Proposition \ref{prop:FT''}, satisfies $h_{Tb}\ge - \mu'$, for some $\mu'>\mu$, it suffices to consider $f(u, v)$ for $u, v\ge -\mu'$. A simple estimate 
\begin{align}\label{f-est} 
f(u, v)\ls 1+|u+v|, \end{align} 
for $u, v\ge -\mu'$, which follows from Lemma \ref{lem:f-low-bnd} below,  $K_{Tb}\ls T +|h_x + h_y|$. This implies the inequality
\[G_{Tb}\ge 
w (T+ |h_x + h_y|)^{-1} w \ge 0.\]

Since the gaps between  the eigenvalues $\lambda_n=b(2n+1)$  of  $-\Delta_{a_b}$ on $\frach$ are equal to $b$, we can choose $m$ s.t. $|\lambda_m-\mu|\ls b$.

 Recall that $L_{Tb}$ acts on the space of the Hilbert-Schmidt operators which can be identified through integral kernels with $\mathfrak{h} \otimes \mathfrak{h}$. 
 Let $\phi_m$ be the normalized eigenfunction of $-\Delta_{a_b}$ corresponding to  the eigenvalues $\lambda_m=b(2m+1)$. We take $u, \|u\|=1$, s.t.  $\phi:=w u=c\phi_m\otimes \phi_m$, where $c=\|w^{-1}(\phi_m\otimes \phi_m)\|^{-1}$ is the normalization constant coming from taking $\|u\|=1$, to obtain
\begin{align}\label{f-est} \lan u, G_{Tb} u\ran&\gs (T+|\lambda_m-\mu|)^{-1}\|w^{-1}(\phi_m\otimes \phi_m)\|^{-2}\notag\\
&\gs (T+b)^{-1}\|w^{-1}(\phi_m\otimes \phi_m)\|^{-2}. \end{align}
Now, 
 write $\phi_m (x)= \sqrt{b}\phi_m^0(\sqrt b x)$, where $\phi_m^0 (x)$ is independent of $b$. Furthermore, by  the assumption on $v$, we have $w\gs |x-y|^{-\kappa/2}, \kappa<2$, for $|x-y|$ sufficiently large. 
This gives $\lan u, G_{Tb} u\ran \gs  (T+b)^{-1} b^{\kappa/2}\ra \infty$ as $ b\ra 0,$ provided $ T\ls b^\sigma, \sigma >\kappa/2$.

Thus we have shown that $\|G_{Tb}\|$, or the largest eigenvalue of $G_{Tb}$, can be made arbitrarily large if $T$ and $b$ are sufficiently small, which, by the BSP, proves Theorem \ref{thm:T-b-small}. 

Finally,  \eqref{f-est} can be proven by analyzing 
$f(u, v)$, $u, v\ge -\mu'$,  separately in several domains. 
  \eqref{f-est} also follows from the stronger statement proved below.   
\end{proof}
Bounds on the function $f(u, v):= \frac{u+ v}{\tanh(u)+\tanh( v)}$ used in the proof above could be proved directly in an elementary way; they also follow from the following.

\begin{lemma}\label{lem:f-low-bnd}
The function $f(u, v):= \frac{u+ v}{\tanh(u)+\tanh( v)}$ has the minimum $1$ achieved at $u = v=0$.
\end{lemma}
\begin{proof} 
To find minimum of $f$, we look for its critical points. We let $g(u,v) = \tanh(u) + \tanh(v)$ and compute
\begin{align}
	\nabla f = \frac{1}{g}( 1- f(u,v) \sech^2(x), 1-f(u,v) \sech^2(y)).
\end{align}
Setting $\nabla f = 0$, we see that
\begin{align}\label{f-inv}
	f(u,v)^{-1}=  \sech^2(u) \ \text{ and }\ 
	f(u,v)^{-1}= \sech^2(v) .
\end{align}
It follows that $\sech^2(u) = \sech^2(v)$ and therefore either $u=v$ or $u=-v$. If $u = -v$, then $f(u,v) = \sech^{-2}(u)$. So critical points are all such $u = -v$ and the  minimum is reached at $u = 0$ and $f(0,0) = 1$. If $u=v$, then \eqref{f-inv} becomes
\begin{align}
	\tanh(u) = u \, \sech^2(u), \text{ or equivalently, }\
	\sinh(u)\cosh(u) = u .
\end{align}
This shows that $\sin(2u) = 2u$ which implies that $u=0$. Hence  minimum is reached at $u = 0$ and $f(0,0) = 1$.
 \end{proof}

 
\section{The existence of the vortex lattices} \label{sec:vort-latt-exist} 
In this section, we prove Theorem \ref{thm:BdGExistence} on existence of the vortex lattice solutions to the BdG equations with an arbitrary vortex (the first Chern) number $n$. Recall that 
we drop self-interaction terms $v* \rho_\gamma$ 
 and $v^\sharp \gamma$. We minimize the resulting energy for $\Tr(\gamma)$ fixed. Hence we omit the term $- \mu \Tr(\gamma)$ in \eqref{FT-def}. Hence, with the notation $h_{a}:= -\Delta_{a}$, the free energy functional $F_{T}(\G, a)$ in \eqref{FT-def} becomes 
\begin{align} \label{cF}
	\F (\G,a) = 
\Tr\big(h_{a} \gamma \big)	 +\frac{1}{2} \Tr\big( \s^* v^\sharp \s \big) 
+& \int dx  |\curl a |^2 
 - TS(\G).
\end{align}

We define the free energy functional $\F (\G, a)$ in \eqref{cF} on $\mathcal{D}^1_\nu \times \frachbvec^1$, if 
  $S(\G) < \infty$. Otherwise we set $\F (\G,a) = \infty$.

 Theorem \ref{thm:BdGExistence} follows from  Theorem \ref{thm:BdG=EL}  and the following  
  \begin{theorem} \label{thm:ExistMin}
Assume that $T > 0$ and $\|v\|_\infty$ is small and specialize $\frach=\frachb$. There exists a finite energy minimizer $(\G_*, a_*) \in \mathcal{D}^1_\nu \times \frachbvec^1$ of the functional $\mathcal{F}(\G,a)$ on the set $\mathcal{D}^1_\nu \times \frachbvec^1$. This minimizer satisfies $0 < \G_*< 1$ and $g(\G_*)$ (see \eqref{S-expr}) is trace class and has the equivariance and the flux quantization properties, \eqref{VL:equiv} and \eqref{mf-quant}. Furthermore, the minimizer $(\G_*,a_*)$ can be chosen to be even, i.e. satisfying \eqref{even}. \end{theorem}
The last statement, that $g(\G_*)$ is trace class, follows from the fact that if $A \geq 0$, then 
the usual trace and the $\frachb \otimes \C^2$ trace are the same.

By combining this result with Theorem \ref{thm:BdG=EL}, 
we obtain Theorem \ref{thm:BdGExistence}. 

\begin{proof}[Proof of Theorem \ref{thm:ExistMin}] Recall that we are dealing only with $\frach=\frachb$.
 We pass from the positive trace class operator $\g$ to the Hilbert-Schmidt one, $\ka:=\sqrt \g$, and from the vector potential $a$ to $e:=a-a_b$. Note that $\g\in I^{1,1}, \g\ge 0 \Longleftrightarrow \ka:=\sqrt \g \in I^{1,2}$.   Furthermore,
 using $a = a_b+e$, $\div a_b = b$ and $\int_{\Omega} =0$, we compute 
\begin{align}
	 \int_{\Omega} |\curl a|^2 =&  
	 \int_{\Omega} |\curl e|^2 + b^2 |\Omd|. 
\end{align}
Consequently, 
 consider, instead of \eqref{cF}, the equivalent functional 
\begin{align} \label{F-def}
	F (\ka, \al, e)& = \mathcal{F}_T(\G, a_b+e)\big|_{\g=\ka^2} - b^2 |\Omd|\\
 \label{F} &=\Tr\big( \ka h_{a_b+e} \ka \big)	 +\frac{1}{2} \Tr\big( \s^* v^\sharp \s \big) 
+ \int  |\curl e|^2  
 - TS(\G)\big|_{\g=\ka^2}
\end{align}
on the space $I^{1,2}\times I^{1,2}\times \frachvec^{1}$ with  the norm $\|(\ka, \al, e)\|_{(1)} := \|\ka \|_{I^{1,2}} + \|\al \|_{ I^{1,2}}+\|e \|_{\frachvec^{1}}$ and with the side conditions $0\le \G\big|_{\g=\ka^2}\le 1$ and $\tr\ka^2=\nu$.
We will keep the notation $\mathcal{D}^1_\nu$ for $I^{1,2}\times I^{1,2}$ with these side conditions.

We will use standard minimization techniques 
proving that $F (\ka, \al, e)$ is coercive and weakly lower semi-continuous, and $\mathcal{D}^1_\nu \times \frachvec^1$ weakly closed. 

\textbf{Part 1: coercivity.} 
The main result of this step is the following proposition:
\begin{proposition} \label{pro:LowerBound}
Let $T > 0$.We have, for  $\ka\in I^{1, 2}$ with $\tr \ka^2=\nu$, 
\begin{align}\label{cF-lower-bnd}
	F (\ka, \al, e) \geq C_1(\|\ka\|_{I^{1,2}}^{2r}/ \nu + \|e\|_{\frachvec^1}) - C_2
\end{align}
for suitable $C_1,C_2 > 0$ and any $r<1$.
\end{proposition}

\begin{proof}[Proof of Proposition \ref{pro:LowerBound}] 
We begin with estimating the entropy term $- T S(\G)$ (cf. \cite{HaiSei}). Recall the expression  \eqref{S-expr} - \eqref{g-s-def}, which we reproduce here 
\begin{align}
\label{S-expr'}	& S(\G) := \Tr(s(\G))= \Tr(g(\G)),\\
 \label{g-s-def'}	&  g(\G): = -\G\ln\G - (1-\G)\ln(1-\G),\ s(\G ):= -2\G\ln\G,
\end{align} 
and define  the relative entropy
\begin{align}\label{RelatEntropy}	
S(A|B) = \Tr(s(A|B)), \quad	& s(A|B) := A(\ln A - \ln B). 
\end{align}
We define and recall the diagonal and off-diagonal operator-matrix $\G_0$ and $\phi$ as
\begin{align}
\label{Gam0Phi} 
\G_0 := \left( \begin{array}{cc} 
	\g & 0 \\  
	0 & \one-\bar\g 
\end{array} \right), \ 	\phi(\beta) :=\left( \begin{array}{cc}
											0 & \beta \\
											\beta^* & 0
										\end{array} \right). 
\end{align}

\begin{proposition}[cf. \cite{FHSS}] \label{FHSSEntropyBnd} We have for $\G = \G_0 + \phi(\alpha)$,
\begin{align}\label{S-relatS}
	 S(\G) =  S(\Gam_0) - S(\G | \Gam_0)\le   S(\Gam_0). \end{align} 
\end{proposition}
\begin{proof}  We note that for $\G := \G_0 + \phi(\alpha)$,
\begin{align}
	 \G \ln \G & 
	 - \G_0 \ln \G_0  
	= \G \ln \G - \G \ln \G_0 + \G \ln \G_0  
		- \G_0 \ln \G_0\\ 
	&= s(\G,\G_0) + (\G -\G_0)\ln \G_0\\ 
	&= s(\G,\G_0) + \phi(\alpha) \ln \G_0. 
\end{align}
 the last term $\phi(\alpha) \ln [\G_0(1-\G_0)^{-1}]$ has zero trace since it is off-diagonal, we have the first equation in \eqref{S-relatS}.

The inequality in \eqref{S-relatS} follows from  Klein's inequality and the fact that $\Tr \G =\Tr \G_0$.\end{proof}
With the definitions  \eqref{S-expr'}- \eqref{g-s-def'} and \eqref{Gam0Phi}, Eq \eqref{S-relatS} implies
\begin{align} \label{SGam0} 
S(\G)\le S(\G_0)  = \Tr(g(\g)). 
 \end{align}

Next, 
 since $\langle e \rangle = 0$ and $\div e = 0$, the Poincar\'e's inequality shows that 
\begin{align}\label{a'H1-est} \|e\|_{H^1}^2\ls \int_{\Omega_\delta} |\curl e|^2.	
\end{align}
Let $a = a_b+e$ and $\g=\ka^2$. Definition \eqref{cF} and  inequalities \eqref{SGam0} and  \eqref{a'H1-est} 
  give, for any $\del>0$ and some $c>0$, 
\begin{align} \label{cF-lwbnd1}
	F (\ka, \al, e) \ge   \del\Tr(h_{a} \gamma) 
	&+ 2 c\|e\|_{H^1}^2 
	+ \frac{1}{2}\Tr(\alpha^* v^\sharp  \alpha) - TS(\g).
\end{align}

  Next, we estimate $\Tr(h_a \gamma)$, with $a = a_b+e$ and $\g=\ka^2$. Using 
$\div e = 0$, we write 
\begin{align}
\notag	\Tr(h_{a} \gamma) &=\Tr(h_{a_b}\gamma) + 2i\Tr(e \cdot \nabla_{a_b} \gamma) + \Tr(|e|^2 \gamma) \\
\label{Tr1-est}		\geq & (1-\epsilon) \Tr(h_{a_b}\gamma) + (1-\epsilon^{-1})\Tr(|e|^2 \gamma), 
\end{align}
for any $\epsilon > 0$.  Let $k(\g):=\Tr(h_{a_b}\gamma)=\|\gamma\|_{I^{1, 1}}$. For the last term, we claim, for any $r\in (0,1)$, the estimate 
\begin{align}\label{e2gam-est}
	0\le \Tr(|e|^2 \gamma)\ls k(\g)^{1-r}  (\tr \gamma)^{r}  \|e\|_{H^1}^2.  
\end{align}
 We prove this estimate for $r=1/2$, which suffices for us.  For general $r\in (0,1)$, see Lemma \ref{lem:e2gam-est} of Appendix \ref{sec:den-est}. Recall the definition $M_b:=\sqrt{h_{a_b}}$.   We use relative bound \eqref{Sob-ineq3} of Appendix \ref{sec:den-est} to find 
\begin{align}\label{e2gam-est'}
	0\le \Tr(|e|^2 \gamma)\ls \||e|^2 M_b^{-s}\|  \|M_b^{s}\ka \|_{I^{2}}  \|\ka\|_{I^{2}} \ls \|e\|_{H^v}^2 \|\ka\|_{I^{s, 2}} \|\ka\|_{I^{0, 2}}\end{align}
with $s>2(1-v)$. The last estimate gives \eqref{e2gam-est} with $r=1/2$.
 
Eq \eqref{e2gam-est}, together with \eqref{Tr1-est}, 
implies, for some constant $C$ (dependent of $\delta$), 
\begin{align}\label{Trhgam-ineq}
	\Tr(h_{a} \gamma)  \geq & (1-\epsilon) \Tr(h_{a_b})\gamma)\notag\\
	& + (1-\epsilon^{-1})C  k(\g)^{1-r}  (\tr \gamma)^{r}  \|e\|_{H^1}^2, 
\end{align}
with  $\e < 1$ and any $r<1$. 
Now, take $\del =c/(C(\epsilon^{-1}-1)k(\g)^{1-r}  (\tr \gamma)^{r})$ and let  $\del':=c\epsilon/(2C)$. 
Inequality  \eqref{cF-lwbnd1} and \eqref{Trhgam-ineq} give
\begin{align} \label{cF-lwbnd2}
	F (\ka, \al, e) 
	\ge 2\del' k(\g)^r - TS(\g) & +  \frac{1}{2}\Tr(\alpha^* v^\sharp  \alpha) 
	+ c \|e\|_{H^1}^2.  
\end{align}

We aim at estimating from below   the following contribution to the r.h.s. of \eqref{cF-lwbnd2} 
\begin{align} \label{E-expr}E(\g)=\del' [\Tr(h_{a_b}\gamma)/ \tr \gamma]^r - TS(\gamma)\end{align}
 Here, recall, $\del':=c\epsilon/(2C)$,   $r<1$ (any) 
  and  $\g\in I^{1, 1}$ with $\tr \gamma=\nu$.    
  Let $ I^{1, 1}_\nu:=\{\g\in I^{1, 1}: \tr \gamma=\nu\}$. We claim that there is a positive constant $C_\nu$ s.t.   \begin{align} \label{cF-lwbnd3}\inf_{\g\in I^{1, 1}_\nu}E(\g) \geq - C_\nu,\end{align} 

We prove  \eqref{cF-lwbnd3}. 
 We use Jensen's (or H\"older's for $h_{a_b}^{r} \gamma^r \gamma^{1- r}$) inequality. Let $e_k$ be an orthonormal eigenbasis of $h$ with eigenvalues $\lam_k$ and write (for $0 < r < 1$ so that $x^r$ is concave)
\begin{align}
	(\Tr(h_{a_b} \gamma)/\Tr\gamma)^r	=& \big( \sum_{k} \lam_k \lan e_k,\gamma e_k \ran/\Tr\gamma \big)^r\\ 
		&\geq  \sum_{k} \lam_k^r \lan e_k,\gamma e_k \ran/\Tr\gamma 
		= \Tr h_{a_b}^{r} /\Tr\gamma\end{align}
where $\lan e_k, \gamma e_k \ran /\Tr \gamma$ is regarded as a probability measure in the application of Jensen's inequality. Hence, we may write
\begin{align}\label{E-low-bnd}
	E(\gamma) \geq & \Tr (h_{a_b}^{r} \gamma)/\nu - TS(\gamma) =:E_r(\gamma).\end{align}
The r.h.s. is convex 
 and therefore any critical point of $E_r(\gamma)$ is a minimum. This is easily determined to be
\begin{align}
	\gamma_{\mu, \beta} = f_{FD}((h_{a_b}^{r} -\mu)/T)
\end{align}
for a $\mu$ such that $\Tr \gamma = \nu$.  $E_r(\g_{\mu, \beta})$ bounded below for $\mu$ and $\beta$ fixed. Hence  \eqref{cF-lwbnd3} follows.  


Eq \eqref{E-low-bnd}, together with estimate \eqref{cF-lwbnd2}, this gives
\begin{align} \label{cF-lwbnd}
	F (\ka, \al, e) \ge 
\del' 
\Tr(h_{a_b}\gamma)^r/ \tr \gamma	 & +   \frac{1}{2}\Tr(\alpha^* v^\sharp  \alpha)   + c \|e\|_{H^1}^2 
  - C_\nu, 
\end{align}
with $\del' =c\epsilon/(2C)$ and any $0<r<1$.
To estimate the second term on the r.h.s. of \eqref{cF-lwbnd}, we 
 bound $\alpha$ by 
  $\gamma$ via the constraint $0 \leq \G \leq 1$ (see \cite{BBCFS} and references therein):

\begin{lemma}\label{lem:al-gam-bnd}
The constraint $0 \leq \G \leq 1$ implies that
\begin{enumerate}
	\item $0 \leq \G(1-\G) \leq 1$.
	\item $\alpha^*\alpha \leq \bar{\gamma}(1-\bar{\gamma})$ and $\alpha \alpha^* \leq \gamma(1-\gamma)$.
	\item $\Tr (M\alpha \alpha^* M^*) \leq \|M\gamma M \|_{1}$ for any operator $M$.
\end{enumerate}
\end{lemma}
\begin{proof}
Since $0 \leq \G \leq 1$, then $0 \leq 1-\G \leq 1$ as well and therefore  
 $0 \leq \G(1-\G) \leq 1$, proving the first claim. 
From \eqref{Gam}, we see that the 1,1-entry of $\G(1-\G)$ is $0 \leq \gamma(1-\gamma) - \alpha \alpha^*$. By considering the 2,2-entry, we have that $\alpha^* \alpha \leq \bar{\gamma}(1-\bar{\gamma})$. Finally, since $1-\gamma \leq 1$, we see that $M\alpha \alpha^* M^* \leq M \gamma M^*$ completes the proof.
\end{proof}

Since $v$ is bounded, this lemma gives
\begin{align}
	|\Tr(\alpha^* v^\sharp \alpha)| \leq \|v\|_{\infty} \Tr(\alpha^* \alpha)  \leq \|v\|_{\infty} \Tr(\gamma).
\end{align}
  This together with 
  inequality \eqref{cF-lwbnd} implies 
 bound \eqref{cF-lower-bnd} of Proposition \ref{pro:LowerBound}. 
  \end{proof}

\textbf{Part 2: 
weak lower semi-continuity.} 

\begin{lemma}\label{lem:Flsc}
The functional $F (\ka, \al, e)$ is weakly lower semi-continuous in $\mathcal{D}^1_\nu \times \frachvec^1$. 
\end{lemma}
\begin{proof} We study the functional $F (\ka, \al, e)$ term by term. 
For the first term on the r.h.s. of \eqref{F}, 
with $a=a_b+e$, we write 
\begin{align} \label{hAgamSplit} 	
	& \Tr((-\Delta_{a}) \gamma) =  \Tr((-\Delta_{a_b}) \gamma) + 2i\Tr(e \cdot \nabla_{a_b} \gamma)+ \Tr(|e|^2 \gamma). 
\end{align}
Since the first term of \eqref{hAgamSplit} satisfies $ \Tr((-\Delta_{a_b})  \gamma)=\| \ka \|_{I^{1, 2}}^2$ 
 and is quadratic in $\ka$, it is $\|\cdot \|_{I^{1,2}}$-weakly lower semi-continuous.

For the second term on the r.h.s of \eqref{hAgamSplit}, 
 we let $e, e'\in \frachvec^1$ and estimate the difference $\Tr(e \cdot \nabla_{a_b} \gamma) - \Tr(e' \cdot \nabla_{a_b} \gamma')$.
 We write 
\begin{align}\label{eq1}
	&  \Tr(e \cdot \nabla_{a_b} \gamma) - \Tr(e' \cdot \nabla_{a_b} \gamma')\notag \\
	&=  \Tr((e-e') \cdot \nabla_{a_b} \gamma) - \Tr(e' \cdot \nabla_{a_b} (\gamma-\gamma')). 
	\end{align}
For the first term on the r.h.s., 
 letting $c:=e-e'$, we claim that 
\begin{align}\label{c-n-gam-est'}|\Tr(c \cdot \nabla_{a_b} \gamma)|  \ls \|c\|_{H^s} \|\gamma\|_{I^{1,1}},\ s<1. \end{align}
 To prove this inequality, we recall that $M_b:= \sqrt{-\Delta_{a_b}}$  and write $\Tr(c \cdot \nabla_{a_b} \gamma)=\Tr(M_b^{-1}c \cdot \nabla_{a_b}M_b^{-1} M_b \gamma M_b)$ and use a standard trace class estimate to obtain $|\Tr(c \cdot \nabla_{a_b} \gamma)| \ls \|M_b^{-1}c \cdot \nabla_{a_b}M_b^{-1} \| \|M_b \gamma M_b\|_{I^{0,1}}$. Next, we use the boundedness of $\nabla_{a_b}M_b^{-1}$, the relative bound  
  $\|M_b^{-1}c  \|  \ls \|c\|_{H^s}, s<1$ (see \eqref{Sob-ineq1} of Appendix \ref{sec:den-est}), and the relation $\|M_b \gamma M_b\|_{I^{0,1}}= \|\gamma\|_{I^{1,1}}$ to find \eqref{c-n-gam-est'}.  (Recall that  $\| \kappa\|_{I^{s,2}}= \|\gamma\|_{I^{s,1}}^{1/2}$.)

Now, keeping in mind that $\g$ and $\g'$ are non-negative, we claim the following estimate for the second term on the r.h.s. of \eqref{eq1} with $c=e'$:
\begin{align}\label{c-n-gam-est}	|\Tr(c \cdot \nabla_{a_b} &(\gamma - \gamma'))|   \notag\\
	&\ls \|c\|_{H^t}( \|\kappa\|_{I^{1,2}} 
	+\|\kappa'\|_{I^{1,2}}) \|\kappa-\kappa'\|_{I^{s,2}} ,\ s, t<1,\end{align}
where $\kappa:=\g^{1/2}$ and  $\kappa:=(\g')^{1/2}$. To prove this estimate,  we write $\g=\kappa^2, \g'={\kappa'}^2$ to expand 
\begin{align}\label{eq2}	
(\gamma - \gamma') & 
=\kappa (\kappa-\kappa')+(\kappa-\kappa')\kappa'.\end{align}  
Now, we use  the boundedness of $\nabla_{a_b}M_b^{-1}$  and the  relative bound $\|M_b^{-s}c  \|  \ls \|c\|_{H^t}, s, t<1$ (see \eqref{Sob-ineq1} of Appendix \ref{sec:den-est}),  and the relations $\|M_b^s \kappa\|_{I^{0,2}}= \|\kappa M_b^s \|_{I^{0,2}}= \| \kappa\|_{I^{s,2}}$, to find
\begin{align}\label{ineq1}	|\Tr(c \cdot \nabla_{a_b} &\kappa (\kappa-\kappa')) | =|\Tr(M_b^{-s}c \cdot \nabla_{a_b}M_b^{-1} M_b\kappa (\kappa-\kappa')M_b^{s})| \notag\\
	&\ls \|c\|_{H^t}\|\kappa\|_{I^{1,2}} \|\kappa-\kappa'\|_{I^{s,2}} ,\ s, t<1.\end{align}
For the second term on the r.h.s. of \eqref{eq2}, we have using
 the  relative bound  $\|M_b^{-1}c \cdot \nabla_{a_b}M_b^{-s} \|  \ls \|c\|_{H^1}, s<1$  (see \eqref{Sob-ineq2} of Appendix \ref{sec:den-est}), to find
\begin{align}\label{ineq2}	|\Tr(c \cdot \nabla_{a_b} &(\kappa-\kappa') \kappa') | =|\Tr(M_b^{-1}c \cdot \nabla_{a_b}M_b^{-s} M_b^{s}(\kappa-\kappa') \kappa')M_b)| \notag\\
	&\ls \|c\|_{H^1}\|\kappa'\|_{I^{1,2}} \|\kappa-\kappa'\|_{I^{s,2}} ,\ s <1.\end{align}
 The last two estimates yield \eqref{c-n-gam-est}.


Applying \eqref{c-n-gam-est'} and \eqref{c-n-gam-est} to the terms on the r.h.s. of \eqref{eq1}, we find, for $3/4 < s<1$,
\begin{align}\label{2nd-term-cont}
	 |\Tr(c \cdot \nabla_{a_b} \gamma) & - \Tr(c' \cdot \nabla_{a_b} \gamma')|\ls  \|c-c'\|_{\vec{\mathfrak{h}}^{s}}  \|\gamma \|_{I^{1, 1}}\notag \\
	&+ \|c'\|_{\vec{\mathfrak{h}}^{s}}  
	( \|\kappa\|_{I^{1,2}} 
	+\|\kappa'\|_{I^{1,2}}) \|\kappa-\kappa'\|_{I^{s,2}} ,\ s, t<1,\end{align}
where $\kappa:=\g^{1/2}$ and  $\kappa:=(\g')^{1/2}$. 

 Now, we use a standard result for Sobolev spaces, $\frachvec^{s'}$ is compactly embedded in $\frachvec^{s}$, for any $s' > s$, and, perhaps, less standard one, that  $I^{s', 2}$ is compactly embedded in $I^{s, 2}$, for any $s' > s$.  
 (One shows the latter fact by passing to the integral kernels and using a standard  Sobolev embedding result.) 

Now, let $\{(\ka_n, \al_n, e_n)\}$ be a weakly convergent sequence in $
\mathcal{D}^1_\nu \times \frachvec^1$ 
 and denote its limit by $(\ka_*, \al_*, a_*)$. Then, by above, it converges strongly in  $\mathcal{D}^s_\nu \times\frachvec^{s},\ s<1$. Hence,  
 we have by \eqref{2nd-term-cont}, 
\begin{align}\label{2nd-term-cont'}
	 |\Tr(e_n \cdot \nabla_{a_b} \gamma_n) & - \Tr(e_* \cdot \nabla_{a_b} \gamma_*)|   \ra 0,\ n\ra \infty. 
	\end{align}
where, as usual, $\g_n=\ka_n^2$ and $\g_*=\ka_*^2$.

Finally, we consider the difference $\Tr(|e|^2 \gamma) - \Tr(|e' |^2 \gamma' )$ due to the last term in \eqref{hAgamSplit}. We decompose
\begin{align}\label{eq3}\Tr(|e|^2 \gamma) -& \Tr(|e' |^2 \gamma' )= \Tr(|e|^2 (\gamma - \gamma' )) + \Tr((|e|^2-|e' |^2)\gamma').\end{align}
For the first term on the r.h.s. we  claim the following estimate
\begin{align}\label{c-n-gam-est''}|\Tr(|e|^2  (\gamma - \gamma' )) | &\ls \|e\|_{H^t}^2 (\|\kappa\|_{I^{1,2}} +\|\kappa'\|_{I^{1,2}}) \|\kappa-\kappa'\|_{I^{s,2}} ,\ s <1,\end{align}
where  $\kappa:=\g^{1/2}$ and  $\kappa':=(\g')^{1/2}$. We use again \eqref{eq2} and
 the  relative bound  $\|M_b^{-s}|e|^2M_b^{-t} \|  \ls \|c\|_{H^r}, s + t>2(1-r),$ (see \eqref{Sob-ineq3} of Appendix \ref{sec:den-est}) to find, similarly to \eqref{ineq1} and \eqref{ineq2}, 
\begin{align}\label{ineq3}|\Tr(|e|^2  (\gamma - \gamma' )) | &=	|\Tr(|e|^2( \kappa (\kappa-\kappa') + (\kappa-\kappa') \kappa') |\\&\le |\Tr(M_b^{-s}|e|^2M_b^{-1} M_b\kappa (\kappa-\kappa')M_b^{s})| \notag\\
&+|\Tr(M_b^{-1}|e|^2M_b^{-s} M_b^{s}(\kappa-\kappa') \kappa' M_b)| \notag\\
	&\ls \|c\|_{H^t} (\|\kappa\|_{I^{1,2}} +\|\kappa'\|_{I^{1,2}}) \|\kappa-\kappa'\|_{I^{s,2}} ,\ s <1,\end{align}
 which gives  \eqref{c-n-gam-est''}. 
 
 Finally, similarly to \eqref{c-n-gam-est'}, we find for the second term on the r.h.s. of \eqref{eq3},
 \begin{align}\label{c-n-gam-est'''}|\Tr((|e|^2-|e' |^2)\gamma')|  \ls \|e - e'\|_{H^s}^2 \|\gamma\|_{I^{1,1}},\ s<1, \end{align} 
  
 Now, Eqs  \eqref{eq3},    \eqref{c-n-gam-est''} and \eqref{c-n-gam-est'''} imply 
\DETAILS{term $\Tr(|e|^2 \gamma)$ on the r.h.s. of \eqref{hAgamSplit}. 
We show \begin{align}\label{e2-gam-est}
	|\Tr(|e|^2  \gamma)| \le   \|e\|_{H^{t}}^2 \|\g\|_{ I^{s,1}}, s>2(1-1/q), t>1/q, q<2.\end{align}
 Firstly,	 we use \eqref{den-def} to obtain $\Tr(|e|^2  \gamma)=\int |e|^2\rho_\g$, which gives $|\Tr(|e|^2  \gamma)|\ls \|e\|_{L^{p}}^2 \|\rho_\g\|_{L^{q}}$. This and the estimate  $\|\rho_\g\|_{L^{q}}\ls \|\gamma\|_{I^{s,1}}, s>2(1-1/q),$  proven in Appendix \ref{sec:den-est} (see  \eqref{den-ineq4}) yield \eqref{e2-gam-est}.  Estimate \eqref{e2-gam-est} implies}
\begin{align}|\Tr(|e_n|^2 \gamma_n) - \Tr(|e_*|^2 \gamma_*)| 
		&\lesssim  
\|e_n\|_{H^1}^2 (\|\kappa_n\|_{I^{1,2}} +\|\kappa_*\|_{I^{1,2}}) \|\kappa_n-\kappa_*\|_{I^{s,2}}\notag\\
		 &		 + \|\gamma_*\|_{I^{1,1}} \|e_n - e_*\|_{H^{s}}^2,\end{align}
for $s< 1$, and therefore, as above the r.h.s. converges to $0$.

The second term on the r.h.s. of \eqref{F} is  quadratic in $\alpha$ and therefore it is continuous in $\mathcal{D}^1_\nu \times \frachvec^1$ since $v\in L^\infty$. It follows that it is weakly lower semi-continuous in $\mathcal{D}^1_\nu \times \frachvec^1$.
\DETAILS{we note that for any $\alpha_1,\alpha_2$,
\begin{align}
	& \Tr(\alpha_1^* v^\sharp  \alpha_1)  - \Tr(\alpha^*_2 v^\sharp  \alpha_2) \\
	&= \Tr((\alpha_1-\alpha_2)^* v^\sharp  \alpha_1)  - \Tr(\alpha^*_2 v^\sharp  (\alpha_1 - \alpha_2))
\end{align}
Since $v$ is bounded, we see that the above terms are bounded by, up to constants,
\begin{align}
	\| \alpha_1 - \alpha_2 \|_2 \|v\|_\infty (\|\alpha_1\|_2+\|\alpha_2\|_2) \lesssim \| \alpha_1 - \alpha_2 \|_2 \|v\|_\infty
\end{align}
where the last inequality follows since $0 \leq \alpha^* \alpha \leq \gamma \leq 1$ in $\mathcal{D}$.} 

The third term on the r.h.s. of \eqref{F}, $\int_{\Omd} |\curl e|^2$, is clearly convex. So its norm lower semi-continuity is equivalent to weak semi-continuity. Since it is clearly $\frachvec^1-$norm continuous, it is $\frachvec^1-$weakly lower semi-continuous.

 Hence all the terms on the r.h.s. of the expression \eqref{F} for $F (\G, a)\equiv F (\ka, \al, a)$, save  $-T S(\G)$, are lower semi-continuous under the convergence indicated and therefore so is $F (\G, a)$. The  lower semi-continuity of the latter term is proven in Lemma \ref{lem:Slsc} of Appendix \ref{sec:entropy}, which completes the proof.\end{proof}

We continue with the proof of Theorem \ref{thm:ExistMin}. With the results above, the proof is standard. Let  $\{(\ka_n, \al_n, e_n)\}$ be a weakly convergent sequence in $
\mathcal{D}^1_\nu \times \frachvec^1$ 
 be a minimizing sequence for $F(\ka, \al, e)\equiv F(\G, e)$. By Proposition \ref{pro:LowerBound}, 
 $\|\G_n\|_{(1)}^r + \|e_n\|_{(1)}, r<1,$ is bounded uniformly in $n$. By Sobolev-type embedding theorems,  $(\G_n, e_n)$ converges strongly in $\cD^s_\nu\times \frachvec^{s}$ for any $s<1$ and by 
  the Banach-Alaoglu theorem, 
   $(\G_n, e_n)$ converges weakly in $\cD^1_\nu\times \frachvec^{1}$. Denote the limit by $(\G_*,e_*)$. Since, by Lemma \ref{lem:Flsc}, $F$ is  lower semi-continuous, we have
\begin{align}
	\liminf_{n \rightarrow \infty} F(\G_n, e_n) \ge F(\G_*,e_*).
\end{align}
Hence, $(\G_*,a_*)$ is indeed a minimizer. Furthermore, if we restrict ourselves to even $(\G, a)$, i.e. to $(\G, a)$ satisfying \eqref{even}, then the minimizer $(\G_*,a_*)$ is also even.

Now, we show 
\begin{lemma}\label{lem:eta*EVs}
$0$ and $1$ are not eigenvalues of $\G_*$.  Consequently, $0 < \G_* < 1$.\end{lemma}
\begin{proof} We use throughout the proof that
 for any $v \in \frachb \times \frachb$ and and operator $A$ on $\frachb$, we have $\Tr AP_v = \lan v, Av\ran$, where $P_v$ is the projection onto $v$, and we  write $\lan v, Sv\ran$ instead of $\Tr SP_v$. 

 We assume for the sake of contradiction that $\G_*$ has the eigenvalue $0$. We observe that if $\G_* x = 0$, then the relation \eqref{Gam-prop} implies $\G_* J\bar x = J\bar x$.  Therefore $1$ is also an eigenvalue. \\ 

\textbf{Case 1: there is some $x \in  \ker \G_*$ such that $\lan x, Sx \ran = 0$.} Below, we omit the argument $e$ in $F(\G, e)$. Define
\begin{align}
	\G' := P_x - J\bar P_x J^* = P_x - P_{J\bar x} \label{eqn:G'-Px-JPJ}.
\end{align}
We see that $\G'$ satisfies the last condition in \eqref{Gam'-cond} since $\Tr(S_1 \G') = \lan x, Sx \ran = 0$. We remark that the other 2 conditions of \eqref{Gam'-cond} are not necessarily satisfied. However, we still can compute the perturbation $F(\G_*+\epsilon \G') - F(\G_*)$ by the continuous functional calculus since $\G'$ is the difference of two eigenprojections of $\G_*$. We compute
\begin{align} \label{eqn:EV-case-1-contradiction}
	& F(\G_*+\epsilon \G') - F(\G_*) 
	\asymp  \epsilon  \ln \epsilon  + O(\epsilon) 
	< 0
\end{align}
since $O(|\epsilon \ln \epsilon|) \gs O(\epsilon)$ for $\epsilon$ small, where $O(\epsilon)$ is the size of all differentiable terms in $F(\G_*+\epsilon \G') - F(\G_*)$ with respect to $\epsilon$. We remark that $F$ is not differentiable with respect to the variation $\G_*+\epsilon \G'$ as $x\ln(x)$ is not differentiable at $0$. This results in the nonlinear $\epsilon \ln \epsilon$ term in \eqref{eqn:EV-case-1-contradiction}. Consequently, \eqref{eqn:EV-case-1-contradiction} contradicts minimality of $\G_*$. We conclude that $\G_*$ has a trivial kernel; hence, it has a trivial $1$-eigenspace. Consequently, $0 < \G_* < 1$.\\

\textbf{Case 2: $\lan x, S x \ran \not= 0$ for all $x \in \ker \G_*$.} We can assume that $\lan x, Sx \ran > 0$ for all $x \in \ker \G_*$, otherwise by a linear combination we are back in Case 1. Let $N := \ker \G_* \oplus \ker (1-\G_*)$. Our first goal is to find some $v_0 \in N^\perp$ such that $\lan v_0, S v_0 \ran < 0$.

First, assume for the sake of contradiction that if $x \in N^\perp$, then $\lan x, Sx \ran = 0$. We claim that $S$ is a bijection between $N$ and $N^\perp$. Our assumption $\lan x, Sx\ran = 0$ for all $x \in N^\perp$ implies that $\lan x, Sy \ran = 0$ for all $x,y \in N^\perp$. Otherwise we can find some $x\in N^\perp$ such that $\lan x, Sx \ran \not=0$ by a linear combination. Consequently, $SN^\perp \subset N$. Since $S^2 = 1$ and $S$ is unitary, we see that $N$ is unitarily isomorphic to $N^\perp$. This proves the claim.

Now, let $x \in N$. We see that $\lan x, Sx \ran > 0$ by assumption of Case 2. The above claim shows that $Sx \in N^\perp$. In particular,
\begin{align}
	\lan x, Sx \ran = 0 .
\end{align}
This is a contradiction. We are lead to conclude that there is some $x \in N^\perp$ with $\lan x, S x \ran \not= 0$.

Now we claim that there is some $x \in N^\perp$ with $\lan x, S x \ran \not= 0$ if and only if there is some $x \in N^\perp$ with $\lan x, Sx \ran < 0$. One direction is obvious. In the other direction, we assume that $x \in N^\perp$ and $\lan x, Sx \ran < 0$ and show that a) $J\mathcal{C}x \in N^\perp$ ($\mathcal{C}$ is the complex conjugation) and b) $\lan J\mathcal{C}x,S(J\mathcal{C}x) \ran > 0$. For a), we note that $N$ is an invariant subspace of $J\mathcal{C} = \mathcal{C} J$ (i.e. $\G_* x = 0$ if and only if $\G_*J\mathcal{C} x = J\mathcal{C} x$ using condition \eqref{Gam-prop}) and $J^* = -J$. Consequently, we see that $N^\perp$ is an invariant subspace of $J\mathcal{C}$ i.e. $J\mathcal{C}N^\perp = N^\perp$. b) follows from the fact $J^*SJ = -S$. The claim at the beginning of the paragraph is now proved.

Consequently, our first goal is achieved: there is some $v_0 \in N^\perp$ with $\lan v_0, Sv_0 \ran < 0$.

Now, let $v \in \ker \G_*$ be arbitrary such that $\lan v, S v \ran > 0$ and fix a $v_0 \in N^\perp$ with $\lan v_0, Sv_0 \ran < 0$. Set
\begin{align}
	\G' =& (P_v - JP_{\bar v}J^*) - \frac{\lan v, S v\ran}{\lan v_0, S v_0 \ran} (P_{v_0} - JP_{\bar v_0}J^*)\\
		=&: A - \frac{\lan v, S v\ran}{\lan v_0, S v_0 \ran} B
\end{align}
By construction, we note that $\G'$ satisfies the third condition of\eqref{Gam'-cond} so that $\G_* + \epsilon \G' \in \mathcal{D}^1_\nu$. Moreover, $A$ satisfy the first condition of \eqref{Gam'-cond}, but not the second while $B$ fulfills the first and second conditions of \eqref{Gam'-cond}. This allows us to differentiate the free energy $F_T$ with respect to the variation $\G_* + \epsilon B$ while computing the variation in $A$ explicitly using the continuous functional calculus. Consequently, we see once more
\begin{align}
	& F(\G_*+\epsilon \G') - F(\G_*) 
	\asymp  \epsilon  \ln \epsilon  + O(\epsilon) \
	< 0 .
\end{align}
This again contradicts the minimality of $\G_*$. We conclude that $\G_*$ has a trivial kernel; hence, it has a trivial $1$-eigenspace.
\end{proof}

Finally, since a minimizing sequence $e_n$ converges to $e_*$ strongly in $\vec{\mathfrak{h}}^s$ for any $s<1$, we have, by the magnetic flux quantization \eqref{mf-quant} for $e_n$, the convergence of $e_n$ to $e_*$ and the Stokes theorem, that $\frac{1}{2\pi} \int_{\Omd} \curl a_* = c_1(\rho)\in \Z$, where, recall, $a_*=a_b +e_*$ and $\Omd$ a fundamental cell of $\lat$. \end{proof} 

\subsection*{Acknowledgments}
The first author is very grateful for Almut Burchard's kind help and suggestions. The second author is grateful to Volker Bach, S\'ebastien  Breteaux, Thomas Chen and J\"urg Fr\"ohlich for enjoyable collaboration and both authors thank Rupert Frank and  Christian Hainzl, for stimulating discussions.
 The authors' research is supported in part by NSERC Grant No. NA7901. During the work on the paper, the authors enjoyed the support of the NCCR SwissMAP.

\appendix

\section{Entropy} \label{sec:entropy}

In this appendix we prove the differentiability and expansion of the entropy functional, which we recall here
\begin{align}
\label{S-def''}	& S(\G) := \Tr(s(\G))= \Tr(g(\G)),\\
\label{g-s-def''}	&  g(\G): = -\G\ln\G - (1-\G)\ln(1-\G),\ s(\G ):= -2\G\ln\G.
\end{align}
 This is used in the next two appendices in order to prove  Theorem \ref{thm:BdG=EL} and Propositions \ref{prop:FT''} and  \ref{prop:FT-expan-order2}.
Let 
$d S(\G  ) \G' :=\partial_\epsilon S(\G  + \epsilon \G') \mid_{\epsilon=0} $. We have

\begin{proposition}\label{prop:S-deriv} Let $\G \in \mathcal{D}^1_\nu$ be such that 
 $g(\G) := -\G \ln \G - (1-\G)\ln(1-\G)$ {\bf is trace class} and $\G'$ satisfy \eqref{Gam'-cond}. Then $S$ is $C^1$ and its derivative is given by
\begin{align}\label{dS}	& d S(\G  ) \G' = \Tr( g'(\G ) \G')= \Tr( s'(\G ) \G').
\end{align}
\end{proposition}
\begin{proof}
By  \eqref{S-def''}, it suffices for us to prove the proposition for $s(\G) = -\G\ln(\G)$.
Denote $\G'':=\G +\epsilon \G'$. We write 
\begin{align}\label{S-AB} S(\G'') - S(\G ) &= -\Tr (\G(\ln\G''-\ln \G) - \e \G'(\ln \G''-\ln \G) - \e \G'\ln \G)\\
\label{S-AB}&=: A+B - \e \Tr (\G'\ln \G).
	\end{align} 
Using the formula $\ln a -\ln b= \int_0^\infty [(b + t)^{-1} - (a + t)^{-1}] d t$ and the second resolvent equation, we compute
\begin{align}
	A :=& -\Tr (\G(\G''-\ln \G)) \notag\\
	&= \int_0^\infty \Tr \{ \G[(\G'' + t)^{-1} - (\G+t)^{-1}]\} dt\notag\\
			&\qquad  = -\int_0^\infty \Tr \{ \G (\G+ t)^{-1} \e\et' (\G''+t)^{-1}\} dt\notag\\
	&\qquad  = -\int_0^\infty \Tr \{ \G (\G+ t)^{-1} \e\et' (\G+t)^{-1}\} dt \notag\\
	&\qquad  -\int_0^\infty \Tr \{ \G (\G+ t)^{-1} \e\et'  (\G+ t)^{-1} \e\et' (\G''+t)^{-1}\} dt.
		 \label{A-comp2}
	\end{align}
Similarly, we have
\begin{align}
	B :=& -\Tr (\e \G'(\G''-\ln \G)) \notag\\
	&= \int_0^\infty \Tr \{ \e \G'[(\G''+ t)^{-1} - (\G+t)^{-1}]\} dt\notag\\
			&\qquad  =  -\int_0^\infty \Tr \{ \e \G' (\G+ t)^{-1} \e\et' (\G''+t)^{-1}\} dt.
			\label{B-comp2}
	\end{align}
Combining the last two relations with \eqref{S-AB}, we find
\begin{align}\label{S-exp2} & S(\G +\epsilon \G')- S(\G )=  \e S_1+  \e^2 R_2, 
\\
\label{S1-orig}&S_1:= -\Tr \G'\ln \G - \int_0^\infty \Tr \{ \G (\G+ t)^{-1} \et' (\G+t)^{-1}\} dt, \\
&R_2:= \int_0^\infty \Tr \{ \G (\G+ t)^{-1} \et'  (\G+ t)^{-1} \et' (\G''+t)^{-1} \notag\\
\label{R2-orig}& \qquad - \G' (\G+ t)^{-1} \et' (\G''+t)^{-1}\} dt.	
	\end{align} 
The estimates below show that the integrals on the r.h.s. converge. Computing the integral $ \int_0^\infty \Tr \{ \G (\G+ t)^{-1} \et' (\G+t)^{-1}\} dt= \int_0^\infty \Tr \{ \G (\G+ t)^{-2} \et' \} dt=\Tr  \et' $ in the expression for $S_1$ and transforming  the expression for $R_2$, 
we obtain
\begin{align}\label{S1}  S_1:=& -\Tr \{ \G'\ln \G + \et' \} , \\ 
\label{R2} R_2:= & -\int_0^\infty \Tr \{ t (\G+ t)^{-1} \et'  (\G+ t)^{-1} \et' (\G''+t)^{-1} \} dt. 
	\end{align}

The proof convergence of \eqref{S1-orig} and \eqref{R2} are similar. We consider the case of \eqref{R2}.
We estimate the integrand on the r.h.s. of \eqref{R2}. 
we have 
\begin{align}& | \Tr \{  \G' (\G+ t)^{-1} \et' (\G''+t)^{-1}\} |\notag\\ 
&\qquad \qquad   \le  \| \et' (\G+t)^{-1}\|_{I^2} \| \et' (\G''+t)^{-1}\|_{I^2} 
\end{align} 
Now, we show that the factors on the r.h.s. are $L^2(dt)$. By the second condition in \eqref{Gam'-cond} on $\et'$, we have 
\begin{align}\| \et' (\et^\#+t)^{-1}\|_{I^2} \le &\| \et (\one -\et^\#) (\et^\#+t)^{-1}\|_{I^2}\notag\\ 
& \le \| \xi^\# (\xi^\# +t)^{-1}\|_{I^2},\notag\end{align}
where $\et^\#$ is either $\et$ or $\et''$ and $\xi^\#:=\et^\# (\one -\et^\#)$. Let $\mu_n$ be the eigenvalues of the operator $\xi^\#:=\et^\# (\one -\et^\#)$. Then we have 
\begin{align}  \label{eqn:1stVar}	&  \|\xi^\# (\xi^\# +t)^{-1}\|_{I^2}^2 =\sum_n  \mu_n^2 (\mu_n+t)^{-2} , 
\end{align}
and therefore 
\begin{align}\int_0^\infty  \| \xi^\# (\xi^\#+t)^{-1}\|_{I^2}^2 dt & =\int_0^\infty   \sum_n  \mu_n^2 (\mu_n+t)^{-2}  dt\notag\\
&=\sum_n  \mu_n = \Tr   \xi^\#.\end{align}  
Since $ \et(\one -\et)$ and $ \et''(\one -\et'')$ are trace class operators, this proves the claim and, with it, the convergence of the integral on the l.h.s.. Similarly, one shows the convergence of the other integrals.

To sum up, we proved the expansion \eqref{S-exp2} with $S_1$ given by \eqref{S1}, which is the same as \eqref{dS},  and $R_2$ bounded as $|R_2|\ls 1$.
In particular, this implies that  $S$ is $C^1$ and its derivative is given by \eqref{dS}.
\end{proof}

\begin{proposition}\label{prop:S-expan-order2} 
$S(\G) := \Tr(g(\G)) $ is $C^3$ at $\G_{Tb}$ w.r.to perturbations $\et'$ satisfying \eqref{Gam'-cond}
 Moreover,  we have
\begin{align}\label{S-expan} 
	S(\G_{Tb}+ \e\G') =& S(\G_{Tb}) + \e S'(\G_{Tb})\G' + \frac12 \e^2 S''(\G',\G')  + O(\epsilon^3) 
 \end{align}
where 
$S'(\G_{Tb})\G':=\Tr(g'(\G_{Tb})\G')$, $S''(\G',\G')$ is a quadratic form,
\begin{align}\label{S''}  	
	S''(\G',\G')=& \frac12  \int_0^\infty \Tr \{  (\G+ t)^{-1} \et'  (\G+ t)^{-1} \et' \} dt
\end{align} 
 and the error term is uniform in $\G'$ and is bounded by $\e^3 \Tr(\G_{Tb}(1-\G_{Tb}))$.  For 
$\G'=\phi(\al)$, the quadratic term becomes 
\begin{align}\label{S''-al} &S''(\G',\G') = -\Tr\left(\bar{\alpha} K_{T b} \alpha \right),\\  & K_{T b}:=\frac{1}{T} \frac{h_{T b}^L+h_{T b}^R}{\tanh(h^L_{T b}/T) + \tanh(h^R_{T b}/T)}, \end{align}
 where $h_{T b}:=h_{\gamma_{Tb}, a_b}$.
\end{proposition}
\begin{proof} For the duration of the proof we omit the subindex $Tb$ in $\G_{Tb}$. Recall \eqref{S-AB} - \eqref{B-comp2} and continuing computing $A$ and $B$ in  \eqref{A-comp2} - \eqref{B-comp2} in the same fashion as in the derivation of these equations, we find
\begin{align}
	A 	=&  \int_0^\infty \Tr \{ \G (\G+ t)^{-1} \e\et' (\G+t)^{-1}\} dt \notag\\
	&- \int_0^\infty \Tr \{ \G (\G+ t)^{-1} \e\et'  (\G+ t)^{-1} \e\et' (\G+t)^{-1}\} dt\notag\\
	&+ \int_0^\infty \Tr \{ \G (\G+ t)^{-1} \e\et'  (\G+ t)^{-1} \e\et' (\G+ t)^{-1} \e\et' (\G''+t)^{-1}\} dt, \label{A-comp3}
	\end{align}
and
\begin{align}
	B 			=&  \int_0^\infty \Tr \{ \e \G' (\G+ t)^{-1} \e\et' (\G+t)^{-1}\} dt \notag\\
	&- \int_0^\infty \Tr \{ \e \G' (\G+ t)^{-1} \e\et'  (\G+ t)^{-1} \e\et' (\G''+t)^{-1}\} dt . \label{B-comp3}
	\end{align}
Combining the last two relations with \eqref{S-AB} and recalling the computation of $S_1$, we find
\begin{align}\label{S-exp3}  S(\G +\epsilon \G')&- S(\G )=  \e S_1+  \e^2 S_2+  \e^3 R_3\\
	S_2:=&- \int_0^\infty \Tr \{ \G (\G+ t)^{-1} \et'  (\G+ t)^{-1} \et' (\G+t)^{-1} \notag\\
	&- \G' (\G+ t)^{-1} \et' (\G+t)^{-1}\} dt,\notag\\
		R_3:= &\int_0^\infty \Tr \{ \G (\G+ t)^{-1} \et'  (\G+ t)^{-1} \et' (\G+ t)^{-1}  \et' (\G''+t)^{-1}\notag\\
	&- \G' (\G+ t)^{-1} \et'  (\G+ t)^{-1} \et' (\G''+t)^{-1}\} dt. 
	\end{align} 
Transforming  the expressions for $S_2$ and $R_3$, we obtain
\begin{align}\label{S2'}  
	S_2=& \int_0^\infty \Tr \{ t (\G+ t)^{-1} \et'  (\G+ t)^{-1} \et' (\G+t)^{-1}\} dt, \\
	\label{R3}	R_3	= &-\int_0^\infty \Tr \{ t (\G+ t)^{-1} \et'  (\G+ t)^{-1} \et' (\G+ t)^{-1}  \et' (\G''+t)^{-1}\} dt.
	\end{align} 
Estimates similar to those done after \eqref{R2} show that the integrals on the r.h.s. converge.
This proves the expansion \eqref{S-exp3} with $S_1$ and $S_2$ given by \eqref{S1}, which is the same as \eqref{dS}, and \eqref{S2'} and  $R_3$ bounded as $|R_3|\ls 1$.
Identifying  the quadratic form $S_2$ with  $S''(\G',\G')$, we arrive at the expansion \eqref{S-expan}.

Before computing $S_2\equiv S''(\G',\G')$, we find a simpler representation for it. Integrating the r.h.s. of  \eqref{S2'} by parts, we find
\begin{align}\label{S''-comp}  	S''=& \int_0^\infty \Tr \{ t (\G+ t)^{-2} \et'  (\G+ t)^{-1} \et' \} dt \notag \\
		=& \int_0^\infty \Tr \{  (\G+ t)^{-1} \et'  (\G+ t)^{-1} \et' \} dt \notag\\
	&- \int_0^\infty \Tr \{ t (\G+ t)^{-1} \et'  (\G+ t)^{-2} \et' \} dt.
	\end{align}
But by the cyclicity of the trace the last integral is equal to the first one and therefore we have \eqref{S''}. Eq \eqref{S-sym} gives 
 \begin{align} 	\label{S''-sym}  S''(\G',\G')=& \frac14  \int_0^\infty \Tr \{  [(\G+ t)^{-1} \et'  (\G+ t)^{-1} \notag\\ &+  (\one - \G+ t)^{-1} \et'  (\one - \G+ t)^{-1}]\et' \} dt,
	\end{align}
	
Now, we use \eqref{S''-sym} to compute to $S''$ for $\G' = \phi(\alpha)$. First,  we recall that $\et=\G_{Tb}$ and observe that for $\G' = \phi(\alpha)$,
\begin{align}
		  \Tr((\G_{Tb}+t)^{-1}\G'(\G_{Tb}+t)^{-1}\G') &= 2\Tr((\g_{Tb}+t)^{-1}\al(\one- \bar\g_{Tb}+t)^{-1}\al) \\
		  &= 2\Tr(((x+t)^{-1}(y+t)^{-1} \alpha) \bar{\alpha} )
\end{align}
where 
$x$ and $y$ are regarded as operators acting on $\alpha$ from the left by multiplying by $\gamma_{Tb}$ and 
from the right, by $1-\bar{\gamma}_{Tb}$. Putting this together with 
 a similar expression for the second term on the r.h.s. of \eqref{S''-sym} and performing the integral in $t$, we obtain 
\begin{align} \label{S''1}	& S''(\G',\G')= -\Tr\left[ \bar{\alpha} K  (\alpha) \right],\\  
& \label{K1} K  :=\frac{\log(x)-\log(y)}{x-y} + \frac{\log(1-x) - \log(1-y)}{(1-x)-(1-y)}, 
\end{align}
 with $x$ acting on the left and $y$ acting on the right. \eqref{K1} can be written as 
 \begin{align}\label{K-expr} 
 K&=- \frac{\log(x^{-1}-1)-\log(y^{-1}-1)}{x-y}. 	
\end{align}

Recalling that $\gamma_{Tb} = g^\sharp(h_{T b}/T)=(1+e^{2 h_{T b}/T})^{-1}$, where $h_{T b}:=h_{\gamma_{Tb}, a_b}$, and therefore $x^{-1}-1=e^{2 h^L_{T b}/T}$ and $y^{-1}-1=e^{- 2 \bar h^R_{T b}/T}$, we see that
\begin{align}
	K=& \frac{1}{T}\frac{h_{T b}^L+\bar h_{T b}^R}{(1+e^{h^L_{T b}/T})^{-1} + (1+e^{\bar h^R_{T b}/T})^{-1}}, \end{align}
which, together 
 the hyperbolic functions identities, $(1+e^{h})^{-1}=\frac12 (1-\tanh h)$ and $(1+e^{-h})^{-1}=\frac12 (1+\tanh h)$,  gives \eqref{S''-al}. \end{proof} 
By the definition of the G\^ateaux derive and the hessian and Proposition \ref{prop:S-expan-order2}, we have

\begin{corollary}\label{cor:S'-S''} 
 We have
\begin{align}\label{S'}
	d S(\G_{Tb})\G':=\Tr(g'(\G_{Tb})\G'), 
 	\end{align} 
 and, for $\G'=\phi(\al)$ and with $h_{T b}:=h_{\gamma_{Tb}, a_b}$, 
\begin{align}\label{S''-K} &S''(\G_{Tb})\phi(\al) = -\phi(K_{T b} \alpha),\ 
 K_{T b}:=\frac{1}{T} \frac{h_{T b}^L+h_{T b}^R}{\tanh(h^L_{T b}/T) + \tanh(h^R_{T b}/T)}. \end{align}
\end{corollary}

Our next result on the entropy is the following

\begin{lemma}\label{lem:Slsc}
The functional $- S(\G)$ is weakly lower semi-continuous in $\mathcal{D}^1_\nu$. 
\end{lemma}
\begin{proof}
We use an idea from \cite{L2} which allows to reduce the problem to a finite-dimensional one. We use \eqref{S-relatS}, to pass from $- S(\G)$ to 
 the relative entropy,  $S(\G | \G_0)$, defined in \eqref{RelatEntropy}, with $\G_0$ of the form \eqref{Gam0Phi}, with $\Tr \gamma_0< \infty$ and s.t. $S(\G_0)<\infty$. 
  By \eqref{S-relatS}, $S(\G | \G_0)\ge 0$. 
Moreover,
\begin{align}\label{S-deco}
	 S(\G) = S(\G_0) - S(\G|\G_0) - \Tr[(\G-\G_0)\ln \G_0].
\end{align}
We choose $\G_0$ so that $(\G-\G_0)\ln \G_0$ is trace class and the term $\Tr[(\G-\G_0)\ln \G_0]$ is wlsc. We take 
\begin{align} \label{eqn:H0choice}
	& \G_0 = f_{\rm FD}(M/T), \, M := \text{diag}(\sqrt{-\Delta_{a_b}}, -\overline{\sqrt{-\Delta_{a_b}}}). 
\end{align}
Since $ f_{\rm FD}(h)=\rbrac{e^{ h}+1}^{-1}$, 
we see that that
\begin{align} \label{eqn:ln-eta0}
	0 \leq -\ln \eta_0 = \ln (1+e^{ M/T}) \lesssim & 1+  M/ T.
\end{align}
This estimate and (\ref{eqn:H0choice}) show that $\Tr[(\G-\G_0)\ln \G_0]$ is $\hat I^1$-norm continuous ($\hat I^1-$norm is defined in \eqref{eta-norm}).  
Indeed, writing $ \ln\eta_0=(1+  M/ T)(1+  M/ T)^{-1} \ln\eta_0$ and using \eqref{eqn:ln-eta0}, we find
\begin{align} \label{eta-ln-eta0-est}	\|(\G-\G_0) \ln\eta_0\|_{I^1} \le \|(\G-\G_0) (1+M/ T)\|_{I^1}&\|(1+  M/ T)^{-1} \ln\eta_0\|\notag\\
& \ls \|\eta-\G_0\|_{I^{1,1}}.
\end{align}
This completes the proof of the claim that $\Tr[(\G-\G_0)\ln \G_0]$ is $\hat I^1$-norm continuous.

Furthermore, since this term is affine in $\G$, it is convex. Thus it is wlsc.


Now, following \cite{L2}, let $s_\lambda(A|B) = \lambda^{-1}(s(\lambda A + (1-\lambda) B) - \lambda s(A)- (1-\lambda) s(B))$ and write
\begin{align}
	S(\G_n | \G_0) +& \Tr(\G_0 - \G_n) 
	= \sup_{\lambda \in (0,1)} \Tr( s_\lambda(\G_n | \G_*)).
\end{align}
Since the entropy function $s$ is concave, $s_\lambda(A|B) \geq 0$ for any $A,B$. For any non-negative operator $T$ on $L^2(\Om)$, $\Tr_{L^2(\Om)} T = \sup_P \Tr_{L^2(\Om)} PT$ where the sup is taken over all finite rank projections. Hence, we may write
\begin{align}
	S(\G_n | \G_0) +& \Tr(\G_0 - \G_n) 
	= \sup_{\lambda \in (0,1)} \sup_{P} \Tr( Ps_\lambda(\G_n | \G_0)
\end{align}
where the $\sup_P$ is taken over all finite rank projections $P$. It follows that for any $\lambda \in (0,1)$ and any finite rank projection $P$,
\begin{align}
	S(\G_n | \G_0) + \Tr(\G_0 - \G_n) \geq \Tr( Ps_\lambda(\G_n | \G_0))
\end{align}
Since $\G_n \rightarrow \G_*$ in $\| \cdot \|_{(0)}$ (hence in operator norm) and $-x\ln x$ is continuous on $[0,1]$, we see that
\begin{align}
	 s_\lambda(\G_n | \G_0) \rightarrow   s_\lambda(\G_* | \G_0) \end{align}
in the operator norm. 
In particular, for any finite dimensional projection $P$, 
\begin{align}
	\Tr(P  s_\lambda (\G_n | \G_0)) \rightarrow \Tr( P   s_\lambda(\G_* | \G_0)).
\end{align}
Consequently,
\begin{align}
	\liminf_{n \rightarrow \infty} S(\G_n | \G_0) +& \Tr(\G_0 - \G_n)\notag \\
	& \geq \Tr( Ps_\lambda(\G_* | \G_0))
\end{align}
Now taking $\sup_{\lambda \in (0,1)}$ and $\sup_P$ and using that $\Tr(\G_0 - \G_n)=0$, by condition \eqref{Gam'-cond}, we see that
\begin{align}
	\liminf_{n \rightarrow \infty} S(\G_n | \G_0) \geq \Tr( s(\G_* | \G_0) \, .
\end{align}
which implies the desired statement. 
  \end{proof}

As an aside not used in this paper, we compute the hessian, $\p_{\g\g}S(\g_{T b})$, of $S$ w.r.to diagonal perturbations,
\begin{align}\label{d-gam'} d(\g'):=\left( \begin{array}{cc} \g' & 0 \\ 0  & -\bar\g' \end{array} \right). 
 \end{align}
 $\p_{\g\g}S(\g_{T b})$ is defined by 
\[\lan\g', \p_{\g\g}S(\g_{T b})\g'\ran:=\p^2_\e S(\G_{Tb}+ \e d(\g'))\big|_{\e=0}.\] 
We have
\begin{proposition}\label{prop:S''-gam} 
 The hessian operator $\p_{\g\g}S(\g_{T b})$ is given by 
\begin{align}\label{S''-gam} \p_{\g\g}S(\g_{T b})= 
\frac1T\frac{h_{T b}^L-h_{T b}^R}{\tanh(h^L_{T b}/T) - \tanh(h^R_{T b}/T)}, \end{align}
 where, recall, $h_{T b}:=h_{\gamma_{Tb}, a_b}$.
\end{proposition}
 \begin{proof}Our starting point in the formula \eqref{S''} and hence we begin with the computation of the term $\int_0^\infty \Tr[(\G+ t)^{-1} \et'  (\G+ t)^{-1} \et'] dt$, with $\G' = d(\g')$, where $d(\g')$ 
 denotes the perturbation in $\gam$ given by
\begin{align}\label{d-gam'} d(\g'):=\left( \begin{array}{cc} \g' & 0 \\ 0  & - \bar\g' \end{array} \right). 
 \end{align}

 First,  we recall that $\et=\G_{Tb}$ and observe that for $\G' = d(\g')$,
\begin{align}
		  \Tr( &(\G_{Tb}+t)^{-1}\G'(\G_{Tb}+t)^{-1}\G') =  \Tr((\g_{Tb}+t)^{-1}\g'(\g_{Tb}+t)^{-1}\g'\\
		  &+ (\one- \bar\g_{Tb}+t)^{-1}\bar \g'(\one- \bar\g_{Tb}+t)^{-1}\bar \g') \\
		  &= \Tr(([(x+t)^{-1}(x'+t)^{-1}  +  (\one-x+t)^{-1} (\one -x'+t)^{-1} ]\g')\g')
\end{align}
where the last follows from $\Tr(A) = \Tr(\bar{A})$ for self-adjoint operators, and 
$x$ and $x'$ are regarded as operators acting on $\g'$ from the left by multiplying by $\gamma_{Tb}$ and 
from the right, by $\gamma_{Tb}$. 
Performing the integral in $t$, we obtain $S''(\G',\G')= -\Tr\left[ \bar{\g'} K'  (\g') \right],$ where the operator $K'$ 
is given by
 \begin{align}
	  K':=\frac{\log(x)-\log(x')}{x-x'}+ \frac{\log(\one -x ) - \log(\one -x')}{(\one -x )-(\one -x' )}, 
\end{align}
 with $x$ acting on the left and $x'$ acting on the right. Clearly, $K'$ is identified with $\p_{\g\g}S(\g_{T b})$. Rewrite the operator  $K'$ as 
 \begin{align}\label{S''-expr}	K'	&= -\frac{\log(x^{-1}-1)-\log({x'}^{-1}-1)}{x-x'} 
\end{align}

Recalling that $\gamma_{Tb} = g^\sharp(h_{T b}/T)=(1+e^{2 h_{T b}/T})^{-1}$, where $h_{T b}:=h_{\gamma_{Tb}, a_b}$, we see that
\begin{align}
	K =& \frac{1}{T}\frac{h_{T b}^L-h_{T b}^R}{(1+e^{h^L_{T b}/T})^{-1} - (1+e^{h^R_{T b}/T})^{-1}}, \end{align}
which, together with \eqref{S''-expr} and the hyperbolic functions identities, $(1+e^{h})^{-1}=\frac12 (1-\tanh h)$ and $(1+e^{-h})^{-1}=\frac12 (1+\tanh h)$,  gives \eqref{S''-gam}. \end{proof}

\section{Energy 
functional: Proof of Theorem \ref{thm:BdG=EL}} \label{sec:energy}

The proof of Theorem \ref{thm:BdG=EL} consists three parts: 1) differentiability of $F_T$, 
 2) identification of the BdG equations with the Euler-Lagrange equation of $F_T$, and 3) showing minimizers of $F_T$ among the set $\mathcal{D}^1_\nu \times \frachbvec^1$ are critical points.

\textbf{Part 1: differentiability.} We consider first the variation $\G  + \epsilon \G'$ for $\epsilon > 0$ small and perturbations 
 satisfying \eqref{Gam'-cond}, 
  Note that such $\G'$ satisfies, for $\epsilon$ small enough,
 \begin{align}
	0 \leq \G  + \epsilon \G' \leq 1.
\end{align}
Let $d_\G F_T(\G , a  ) \G' :=\partial_\epsilon F_T(\G  + \epsilon \G' , a ) \mid_{\epsilon=0} $, if the r.h.s.  exists. From \eqref{energy}, it is easy to see that $E(\G, a)$ is Fr\'echet differentiable and 
\begin{align}\label{dcF} 
	d_\G E(\G , a  ) \G' =  \Tr(\Lambda(\G , a ) \G'). 
\end{align}
Hence it suffices to prove the  Fr\'echet differentiability of $S(\G  )$. This is done in Appendix \ref{sec:entropy} above.

 Differentiability of $F_T$ with respect to $a$ is standard and can be done much easily. The only two terms in $F_T$ that depend on $a$ are $\Tr((-\Delta_a)\g)$ and $\frac{1}{2}\int |\curl a|^2$. The first term can be differentiated by using $-\Delta_{a_0+a'} = (-\Delta_{a_0})^2 -2 a' (-i\nabla-a_0) + |a'|^2$ while the second term is differentiable by standard variational calculus. Hence the differentiability of $F_T$ follows from \eqref{FT-def}, \eqref{dcF} and Proposition \ref{prop:S-deriv}.																		
\medskip

\textbf{Part 2: Euler-Lagrange equation.} 

Now, we show that if $0 < \G < 1$ and  $d_\G F_T(\G,a)\G' = 0$ and $d_a F_T(\G,a)a' = 0$ for all $\G'$ on $\frachb \times \frachb$ satisfying 
\eqref{Gam'-cond} and $a' \in \frachvec^1$, then $(\G,a)$ satisfies the BdG equations \eqref{BdG-eq-t} - \eqref{Amp-Maxw-eq}. To fix ideas, we consider only the space $\frach=\frachb$ and  indicate the proper modification  in the proof for $\frach=L^2(\R^2)$ when necessary.

We start with $d_{\G} F_T\G'=0$ for all $\G'$ satisfying \eqref{Gam'-cond}. 
First, we construct explicitly a dense subset of perturbations $\G'$ satisfying \eqref{Gam'-cond}.
For a critical point $(\G,a),  0< \G < 1$, we define a reference unit vector 
$v_0 = (1,0)^T \in \frachb \times \frachb$. In the case of $\frach= L^2(\R^2)$, take $v_0 = (f,0)$ for any unit norm $f \in L^2(\R^3)$ supported in $\Om$.
We 
note that the difference in norm of $v_0$'s two components is simply
\begin{align}
	0 \not= 1 = \|(v_0)_1\|_{\frach}^2 - \|(v_0)_2\|_{\frach}^2 = \lan v_0, S v_0 \ran = \Tr(SP_{v_0})
\end{align}
Note that, in the case of $L^2(\R^2)$ we use the trace per volume.  However, since we chose $f$ to be supported in $\Om$, the same expression above holds.

For simplicity and without loss of generality, we assume that $v_0$ is in the image of $\G(1-\G)$ since $0 < \G < 1$ (i.e. its range is dense). We define $V \subset \frachb \times \frachb$ as
\begin{align}
	V = \{ v :& \, \|v\|_2=1, \, v = \G(1-\G) \xi,\ \xi \in \frachb \times \frachb \} \label{eqn:def-v}
\end{align}
For each $v\in V$, we define
\begin{align}
	\G'_v = (P_v - P_{J \bar v}) - \frac{\Tr SP_v}{\Tr SP_{v_0}} (P_{v_0} - P_{J \bar v_0})\, .
\end{align}
where $P_x$ is the orthogonal projection onto $x$ and $J$ is the complex structure in \eqref{Gam-prop}. 
\begin{lemma} $\G'_v$ satisfies \eqref{Gam'-cond}. \end{lemma} 
\begin{proof} To prove the first condition of \eqref{Gam'-cond}, we only prove it for $P_v - P_{J \bar v} = P_v - J P_{\bar v} J^*$ since this condition is real linear. (Note that $S$ is self-adjoint, so $\lan v, S v\ran =\Tr SP_v$ is real for all $v$.) We note that $J^*=J^{-1} = -J$ and has only real number components. Let $\mathcal{C}$ denote the complex conjugation. It follows then that
\begin{align}
	J^*(P_v - JP_{\bar v}J^*)J =& J^*P_v J - P_{\bar v} \\
		=& -\mathcal{C}(P_v - J^* \bar P_{v} J)\mathcal{C} \\
		=& -\mathcal{C}(P_v - JP_{\bar v}J^*)\mathcal{C} \, .
\end{align}
This prove the first condition in \eqref{Gam'-cond}.

To prove the second condition in \eqref{Gam'-cond}, it suffices to show that $P_v$ satisfies this condition for every $v \in \frachb \times \frachb$ since $J$ is unitary. For any $v = \G (1-\G)\xi, x \in \frachb \times \frachb$, we note that
\begin{align}
	\|P_v x\| = |\lan v, x\ran| = |\lan \G(1-\G)\xi, x \ran | \ls \|\xi\|_2\|\G(1-\G)x\|_2 .
\end{align}
This shows that $(\G')^2 \leq C [\G (1-\G )]^2$. This proves the second condition in \eqref{Gam'-cond}.

Finally, we prove the last condition in \eqref{Gam'-cond}. For any unit norm $v \in \frachb \times \frachb$,
\begin{align}
	\Tr S_1(P_v - JP_{\bar v}J^*) =& \Tr S_1P_v - \Tr S_1JP_{\bar v} J^* \\
		=& \Tr SP_v - \Tr J^*SJP_{\bar v} .
\end{align}
We note that $J^*S_1J = S_2 := \text{diag}(0,1)$. Hence,
\begin{align}
	\Tr S_1(P_v - JP_{\bar v}J^*) =& \Tr SP_v 
\end{align}
It follows that
\begin{align}
	\Tr S_1 \G'_v = \Tr SP_v - \frac{\Tr SP_v}{\Tr SP_{v_0}} \Tr SP_{v_0} = 0.
\end{align}
This proves that  $\G'_v$ satisfies the last condition in \eqref{Gam'-cond}. \end{proof}

We show that  $0 < \G < 1$ and  $d_\G F_T(\G,a)\G' = 0$ for all $\G'$  satisfying \eqref{Gam'-cond} imply $\Lambda(\G ,a ) - \mu S -Tg'(\G ) = 0$ for some $\mu$ and $S = \text{diag}(1,-1)$ (see Proposition \ref{prop:stat-sol}). 
First note that equations \eqref{FT-def}, \eqref{dcF}, \eqref{dS} and \eqref{S-sym} yield that
\begin{align}\label{detaFT}
	d_\G F_T(\G,a) \G'= \Tr \left[A\G' \right].
\end{align}
where $A := \Lambda(\G ,a ) -Tg'(\G )$. If $(\G ,a )$ is  a critical point, then, for all $v \in V$, it satisfies
\begin{align}
	\tr(A\G'_v) = 0 \, .
\end{align}
 Since $A$ is in the tangent space of all the $\G$ such that $J^*\G J = 1-\bar \G$. We note that $A$ also satisfies  the first condition in \eqref{Gam'-cond}. It follows that
\begin{align}
	0 =& \Tr(A \G'_v) = \Tr(A P_v) - \Tr(AJ\bar P_v J^*) \\
			&- \frac{\Tr SP_v}{\Tr SP_{v_0}}  (\Tr(A P_v) - \Tr(AJ\bar P_v J^*)) \\
		=& \Tr(A P_v) - \Tr(J^*AJ\bar P_v ) \\
			&- \frac{\Tr SP_v}{\Tr SP_{v_0}}  (\Tr(A P_{v_0}) - \Tr(J^*AJ\bar P_{v_0} )) \\
		=& \Tr(A P_v) + \Tr(\bar A \bar P_v ) \\
			&- \frac{\Tr SP_v}{\Tr SP_{v_0}}  (\Tr(A P_{v_0}) + \Tr(\bar A \bar P_{v_0} )) \\
		=& 2\Tr(A P_v) - \frac{\Tr SP_v}{\Tr SP_{v_0}} (2\Tr AP_{v_0})
\end{align}
We conclude that
\begin{align}
	\Tr AP_v = \frac{\Tr AP_{v_0}}{\Tr SP_{v_0}} \Tr SP_v =: \mu \Tr SP_v
\end{align}
for all $v \in V$. We note that $\mu$ is real since $A$ and $S$ are self-adjoint. Since $0 < \G < 1$, the linear space spanned $V$ is dense. We conclude that $A$ is a multiple of $S$, which we denote by $\mu$. This shows that
\begin{align}
	0 = A - \mu S 
\end{align}
We conclude that $(\G,a)$ solves the first BdG equation, \eqref{Gam-eq}.

Now, we consider the equation $d_a F_T(\G,a)a' = 0$. As was mentioned above, one can easily show that 
\begin{align}
	d_aF_T(\eta, a)a' = 2\int dx v \cdot a' ,
\end{align}
where $v:=\curl^* \curl a - j(\gamma, a)$ and the perturbation $a' \in \frachvec^1$ is divergence free and mean zero. Hence, to conclude that $v=0$, 
we have to show that $v:=\curl^* \curl a - j(\gamma, a)$ is divergence free and mean zero. 
 (Indeed, any vector field $e$ can be written as $e=a'+\n g +c$, where $a'$ is divergence free and mean zero and $c$ is a constant, and therefore $\int_\Om v e=\int_\Om v a'- \int_\Om \divv v g+ c\int_\Om v $. So if  $v$ is divergence free and mean zero and $\int_\Om v a'=0$ for every $a'$ divergence free and mean zero, then $v=0$.)

 Clearly, the term $\curl^* \curl a$ is divergence free and mean zero.  So we show that $j(\gamma, a)$ is divergence free and mean zero. For the first property, we use the fact that our free energy functional is invariant under gauge transformation. In fact, it suffices to use the gauge invariance of the first line in \eqref{energy}, $E_1(T^{\rm gauge}_{t\chi}(\G, a))=E_1(\G, a)$, where  
 $E_1(\G, a) :=\Tr\big((-\Delta_a) \gamma \big) +\Tr\big((v* \rho_\g) \gamma \big) -\frac12\Tr\big((v^\sharp \g) \gamma \big)$. It gives
\begin{align}
	0 = \partial_t \mid_{t=0} E_1(T^{\rm gauge}_{t\chi}(\eta, a))
\end{align}
for all $\chi \in H^1_{\rm loc}$ 
which are $\mathcal{L}$-periodic. Using the cyclicity of trace, 
we compute this explicitly 
\begin{align}\label{J-comp'}	0 
		=& \Tr(\Re(2i\nabla_a \gamma) \cdot \nabla \chi) + \Tr([\gamma, h_{\gamma,a}] \chi).
\end{align}
 Since $(\G,a)$ solves the BdG equation, we have that $[\Lambda(\G,a), \G] = 0$. Taking the upper left component of this operator-valued matrix equation, we see that
\begin{align}
	[h_{\gamma,a}, \gamma] + (v^\sharp \alpha) \bar{\alpha} - \alpha (v^\sharp\bar{\alpha}) = 0
\end{align}
Since $v(x)=v(-x)$, we conclude that the integral kernel of $(v^\sharp \alpha) \bar{\alpha} - \alpha (v^\sharp\bar{\alpha})$,
\begin{align}
	\int (v(x-z)-v(z-y))\alpha(x,z)\bar{\alpha}(z,y) dz,
\end{align} 
 is zero on the diagonal. Thus, the same conclusion holds for $[\gamma, h_{\gamma,a}]$. Consequently, $\Tr([\gamma, h_{\gamma,a}] \chi)=0$ and we conclude, by \eqref{J-comp'}, that 
\begin{align}\label{J-comp}
	0 = -\Tr(\Re(2i\nabla_a \g) \cdot \nabla \chi) =
	 \int_{\Omd} j(\gamma,a) \cdot \nabla \chi=0.
\end{align}
Since this is true for every $\chi \in H^1_{\rm loc}$ which are $\mathcal{L}$-periodic, it follows that $\div j(\gamma, a) = 0$. 

To show that $j(\gamma, a)$ is mean zero, 
we use, that by our assumptions, $\g$ is even and $a$ is odd. Since for any operator $A$, $u^{\rm refl}\den[A] = \den[u^{\rm refl}Au^{\rm refl}]$ where $(u^{\rm refl}f)(x) = f(-x)$, his shows that $j(\gamma, a)$ is odd. Hence so is $v:=\curl^* \curl a - j(\gamma, a)$ and therefore $v:=\curl^* \curl a - j(\gamma, a)=0$.

Since $\div a = 0$, we may replace $\curl^* \curl$ by $-\Delta$. Hence, the elliptic regularity theory shows that $a \in \frachbvec^2$. This completes the proof. $\Box$

\medskip

\textbf{Part 3: minimizers are critical points.} 
For a minimizer $(\G,a)$, we have that $d_\G F_T(\G,a)\G'$, $d_a F_T(\G,a)a' \geq 0$. 
 Since $\frachvec^1$ is linear, $a' \in \frachvec^1$ if and only if $-a' \in \frachvec^1$. So $d_a F_T(\G,a)a' = 0$ for all $a \in \frachbvec^1$. Similarly, we note that $\G'$ satisfies the assumption \eqref{Gam'-cond} 
 if and only if $-\G'$ satisfies the same requirement. Hence we conclude that 
	$0 = d F_T(\G  , a ) \G' ,$ 
which completes the proof. $\Box$

\section{
 Another proof of the first part of Proposition \ref{prop:gam-mt-invar-J0}} \label{sec:mt-den}  
 We begin with some general definitions. 
For a given lattice $\lat$, we fix a fundamental cell $\Om$ and define an inner product on $L^2_{loc}(\R^2)$ by
\begin{align}
	\lan f, g \ran_{\cH} = \sum_{n=1}^\infty 2^{- n} (2n-1)^{-2} 
	\int_{D_n} \bar f(x) g(x) dx=: \int_{\R^2} \bar f g d\mu
\end{align}
 where $D_n$ is the $(2n-1) \times (2n-1)$ block of $\Om$'s centred at the origin and
\begin{align}
	d\mu(x) = \sum_{n=1}^\infty 2^{- n} (2n-1)^{-2} \chi_{D_n} dx =: m(x) dx
\end{align}
where $dx$ is the usual Lebesgue measure on $\R^2$. Note that $m(x) > 0$ always. We construct
\begin{align}
	\cH = \text{closure}\{ f\in L_{loc}^2(\R^2): \|f\|_{\cH} < \infty \}
\end{align}
The upshots are: 

1) 
$\lan f, g \ran_{\cH}$ is an inner product and therefore $\cH$ is a Hilbert space;

2)  if $f$ and $g$ are $\lat$-gauge periodic, then 
\begin{align}
	\lan f, g \ran_{\cH} = 
	\int_{\Om} \bar f g dx
\end{align}
and therefore, $\frachb$ isometrically embeds in $\cH$; 

3) magnetic translations leave $\cH$ invariant and hence $\tau_{bs}(A)$ is well defined on $\cH$; 

4) an operator $A$ on $\cH$ satisfying $\tau_{b h} A = A$ leaves the subspace $\frachb$ invariant.

The last property is important for us since we are interested in operators on $\frachb$ which are characterized by the property that $\tau_{b h} A = A$.

For a bounded operator $A$ on $\cH$,  we identify the integral kernel, $A'$, as an element of $\cH \otimes \cH$ satisfying the relation 
\begin{align}\label{A-A'-rel}
	\lan g\otimes \bar f, A'\ran_{\cH \otimes \cH}=\lan g, A f\ran_{\cH}. 
\end{align} 	
In this case, we have
\begin{align}\label{A-A'}
	(Af)(x) = \int_{\R^2}  A'(x,y) f(y) d\mu(y)
\end{align}
which for $A'$ and $f$  $\lat$-gauge periodic reduces to
\begin{align}
	(Af)(x) = 
	\int_{\Om}  A'(x,y) f(y) dy
\end{align}

Next, for a locally trace class operator $A$ on $\cH$, we define the function (density) $\den[A](x)$ by the relation
\begin{align}\label{den-def'}
	\int f \den[A] dx &:=\Tr_\cH( f A),\ \quad \forall f\in C_0^\infty .
\end{align} 
If $A'(x,y)$ is continuous, then $\den[A](x):=A'(x,x)$. 

 Now, we are ready for 
\begin{lemma} \label{lem:perA-den}
If an operator $A$ on $\cH$ satisfies $\tau_{b h} A = A$, then  den$[A]$ is constant.
\end{lemma}
\begin{proof}
  By \eqref{den-def'} and the equation $\tau_{b h} A = A$, the definition $\tau_{b h}(A)=u_{b h} A u_{b h}^{-1}$ and the cyclicity of the trace on $\cH$ , we have, for any function  $f\in C_0^\infty $, 
\begin{align}\label{den-mtA}
\int  f \den[A] dx 
	&=\Tr_\cH (  f A)=\Tr_\cH( f \tau_{b h}(A))\\
	&=\Tr_{\cH}( f u_{b h} A u_{b h}^{-1}) = \Tr_{\cH}( u_{b h}^{-1} f u_{b h} A).	
\end{align} 
Using  \eqref{den-mtA}   and the relation $u_{b h}^{-1} f u_{b h}=({u_{h}^{\rm transl}}^{-1}f)$,  we find furthermore 
\begin{align}\label{den-comp'}
	\int  f \den[A] dx  &=\Tr_{\cH}(  ({u_{h}^{\rm transl}}^{-1}f)A)\\ 
 &=\int  ({u_{h}^{\rm transl}}^{-1}f) \den[A] dx \\
 &=\int  f {u_{h}^{\rm transl}}\den[A] dx 	
\end{align} 
This implies $\den[A] = {u_{h}^{\rm transl}}\den[A]$ for all $h \in \R^2$ and therefore is independent of $x$ as claimed.
\end{proof} 

\section{Proof of the existence of solution to \eqref{xi-fp}} \label{sec:xi-fp}

We define $f_T:\R \rightarrow \R$ by
\begin{align}
	f_T(\xi) := v*\den[f_{FD}((-\Delta_{a_b} - \mu + \xi)/T)]
\end{align}
We can derive a more explicit formula for $f_T$.

\begin{lemma} \label{Normal_Explicit_exp}
For each $T > 0$ and $b = \frac{2\pi n}{|\Om|} << 1$,
\begin{align}
	|f_T(\xi) - T B \int_{(\xi-\mu)/T}^\infty f_{FD}(y) dy | \leq C b^2 \label{Normal_Explicit_eqn}
\end{align}
where $B = \frac{\hat{v}(0)}{4\pi }$ and $C$ is independent of $T$.
\end{lemma}
\begin{proof}
By Proposition \ref{prop:mt-inv-al0}, $\den[f_{FD}((-\Delta_{a_b} - \mu + \xi)/T)]$ is constant. It follows that
\begin{align} \label{eqn:f-T-prelim}
	f_T(\xi) = \hat{v}(0)\frac{1}{|\Om|}\Tr f_{FD}((-\Delta_{a_b} - \mu + \xi)/T)
\end{align}
Using the eigenbasis $\psi_{m,j}$ of $-\Delta_{a_b}$, with the eigenvalues $b(2m+1),\ j \in \{1,...,n\}$, we have 
\begin{align}
	 & \Tr f_{FD}((-\Delta_{a_b} - \mu + \xi)/T) \\
	=& \sum_{m \geq 0, j=0,...,n-1} f_{FD}((b(2m+1) - \mu + \xi)/T) .
\end{align}
We note that this sum is a Riemann sum where the mid-point of slice in the Riemann sum calculation is sampled. Thus for $b$ small,
\begin{align}
	& \Tr f_{FD}((-\Delta_{a_b} - \mu + \xi)/T) \\
		\asymp & \frac{n}{2b}\int_0^\infty f_{FD}((x-\mu+\xi)/T) dx + O(b^2)\\
		=& \frac{Tn}{2b}\int_{(\xi-\mu)/T}^\infty f_{FD}(y) dy + O(b^2)\label{Normal_summand}
\end{align}
Recalling that $b = \frac{2\pi n}{|\Om|}$ and using \eqref{eqn:f-T-prelim}, we see that this gives \eqref{Normal_Explicit_eqn}.
\end{proof}

Now we use fixed the point theorem to show existence of solution. To this end, we need an estimate on $f_T(a)-f_T(b)$ so that we can use the Banach contraction mapping principle.

\begin{lemma} \label{lem:Contraction}
Assume that $b \ll 1$. Then there is $C>0$  independent of $T$ and $\delta$ s.t.
\begin{align}\label{fT-Lip}
	|f_T(a) - f_T(b) | < C|T\hat{v}(0)||a-b|.
\end{align}
\end{lemma}
\begin{proof}
Using the same method as Lemma \ref{Normal_Explicit_exp}. We see that
\begin{align}
	 f_T(a) - f_T(b) 
	&= \frac{\hat v(0)}{|\Om|} \sum_{m \geq 0} g^\sharp((b(2m+1) - \mu + a)/T) \notag\\
	&- g^\sharp((b(2m+1) - \mu + b)/T),
\end{align}	
which, by the mean value theorem and the fact that the resulting expression is a Riemann sum of a $L^1$ function, gives
\eqref{fT-Lip}. 
\end{proof}

Thus, we see that $f_T$ is a contraction on $\R$ if $|T\hat{v}(0)|$ is small. Hence, we conclude that it has an unique fixed point on $\R$.

Now we carry out the proof of Lemma \ref{lem:ExpandOfXi}. To see the first claim $\xi < 0$, we only need to note that $\text{sgn} B = \text{sgn} \hat{v}(0) < 0$ while $f_{FD} > 0$ (thus the integral in \ref{Normal_summand} is positive). Hence $f_T(\xi) < 0$ for all $\xi$. So a fixed point of $f_T$ must also be negative.

We use Lemma \ref{Normal_Explicit_exp}. We see that
\begin{align}
	|\xi| \leq & T|B| \int_{\frac{\xi-\mu}{T}}^\infty f_{FD}(y) dy + O(b^2) \\
		\leq & T|B| \int_{\frac{\xi-\mu}{T}}^0 f_{FD}(y) dy) + T|B|\int_{0}^\infty f_{FD}(y) dy + O(b^2) \\
		\leq & T|B| \frac{\mu-\xi}{T} + O(T+b^2) \\
		\leq & |B|(\mu-\xi) + O(T+\delta^2)
\end{align}
Since $|B| = \frac{\hat{v}(0)}{4\pi} << 1$, we see that $\xi$ is bounded.

Now, we explicitly integrate the expression in (\ref{Normal_Explicit_eqn}). We see that
\begin{align}
	\int_{\frac{\xi-\mu}{T}}^\infty f_{FD}(y) dy = \frac{1}{2} \left(\log(e^{\frac{2(\xi-\mu)}{T}} + 1) - \frac{2(\xi-\mu)}{T} \right)
\end{align}
 Expanding the above expression for $T$ small, we see that 
\begin{align} \label{NV-xi_recursion}
	\xi =& -B(\xi-\mu) + \frac{TB}{2} e^{\frac{2(\xi-\mu)}{T}} + O(TBe^{\frac{4(\xi-\mu)}{T}} + b^2) .
\end{align}
Solving for $\xi$ in lowest order $T$ and re-inserting into \eqref{NV-xi_recursion}, we see that
\begin{align}
	\xi =& \frac{B\mu}{1+B} + \frac{TB}{2} e^{-\frac{2\mu}{(1+B)T}} + O(TBe^{-\frac{4\mu}{(1+B)T}} + b^2) .
\end{align}
This gives expansion \eqref{eqn:xiT-expansion}.

Finally, we show that $\xi(T)$ is smooth. This can be seem by noticing that $f_T(\xi)$ is smooth in $\xi$ and $T$ for $T > 0$ and apply the implicit function theorem. The requirement $\partial_\xi f_T(\xi) \not= 0$ (for $T$ small) can be seem by Lemma \ref{Normal_Explicit_exp} and the fact $g^\sharp > 0$ always.


\section{Relative bounds and estimates on density} \label{sec:den-est}
 In this appendix we prove bounds on functions relative to the operator $M_b$ and estimates on density $\rho_\gamma$. Our first result is the following
\begin{lemma} \label{lem:Sob-est} We have the following Sobolev-type  inequalities
 \begin{align} \label{Sob-ineq1} &\|M_b^{-s}c M_b^{-t} \|  \ls \|c\|_{H^r}, s+t>1-r, \\
 \label{Sob-ineq2} &\|\|M_b^{-1}c \cdot \nabla_{a_b}M_b^{-t} \|  \ls \|c\|_{H^r}, t>1-r,\\ 
  \label{Sob-ineq3}&\|M_b^{-s}|e|^2M_b^{-t} \|  \ls \|e\|_{H^r}^2, s+t>2(1-r).
 \end{align}
where in the second estimate we assumed $\divv c=0$. 
 \end{lemma}
\begin{proof} We use the diamagnetic inequality $|M_b^{-s} f|\le M_{b=0}^{-s}|f|$ (see \cite{AHS}) to reduce the problem to the $b=0$ case. To estimate the r.h.s. we write $M_0^{-s}$ as the convolution, $M_0^{-s}u=G_s* u$, where $G_s(x)$ is the Fourier transform of $(1+|k|^2)^{-s/2}$, and use that $G_s(x)$ decays exponentially at infinity and has the singularity $\asymp |x|^{-2+s}$ at the origin. Hence $G_s\in L^t(\R^2), t<2/(2-s)$ and we can estimate by the Young inequality $\|G_s* u\|_{L^k}\ls \|G_s\|_{L^t}\|u\|_{L^q}, 1+1/k=1/t+1/q, t<2/(2-s),$ to obtain
\begin{align} \label{Sob-ineq5} &\|M_b^{-s}  f\|_{L^k}  \ls  \|f\|_{L^r}, 1/k+s/2>1/r, s<2.
 \end{align}

Now, to prove  \eqref{Sob-ineq1} and \eqref{Sob-ineq3}, we apply \eqref{Sob-ineq5} twice and the Young inequality to obtain
\begin{align} \label{Sob-ineq6} &\|M_b^{-s}c M_b^{-t} f\|_{L^2}  \ls  \|c M_b^{-t} f\|_{L^q} \ls  \|c\|_{L^p} \|M_b^{-t} f\|_{L^k} \ls  \|c\|_{L^p} \| f\|_{L^2}, 
 \end{align}
with $1/2+s/2>1/q=1/p +  1/k, 1/k+t/2=1/2$, which implies $s+t>2/p$. To obtain  \eqref{Sob-ineq1} and \eqref{Sob-ineq3}, we use the Sobolev inequalities $\|c\|_{L^p}  \ls  \|c\|_{H^r}, r>1-2/p$ and $\||e|^2\|_{L^p} =\|e\|_{L^{2p}}^2 \ls  \|e\|_{H^r}^2, r>1-1/p$, respectively.

To prove \eqref{Sob-ineq2}, we use, in addition,  $M_b^{-1}c \cdot \nabla_{a_b}M_b^{-s}=M_b^{-1}\nabla_{a_b} \cdot c M_b^{-s}+M_b^{-1} (\n c) M_b^{-s}$ and $\divv c=0$ to reduce the problem to \eqref{Sob-ineq1} with $s=0$. \end{proof}

\begin{lemma} \label{lem:den-est} Let $\g$ be a trace-class and positive operator and let $\ka:=\sqrt\g$. Then 
 \begin{align} \label{den-ineq1} &\|\rho_\gamma\|_{W^{s, 1}} \ls  \|\ka\|_{I^{s,2}} \|\ka\|_{I^{0,2}},  \\ 
 \label{den-ineq2} &\|\rho_\gamma\|_{W^{1, 1}} \ls ( \|\gamma\|_{I^{1, 1}}  \tr \gamma)^{1/2},\\ 
 \label{den-ineq3}&\|\rho_\gamma\|_{L^q} \ls  \|\gamma\|_{I^{1,1}}^{1-r}  (\tr \gamma)^{r},\ \forall r\in (0, 1),\\ 
 \label{den-ineq4}&\| \rho_{\g}\|_{L^{q}}\ls  \|\gamma\|_{I^{s,1}},\ s>2(1-1/q). 
 \end{align}
 \end{lemma}
\begin{proof}
 We use \eqref{den-def} and $\p\rho_{\g}=\rho_{[\p_{a_b}, \g]}$ to obtain $\|\p\rho_\gamma\|_{L^{ 1}}=\sup_{\|f\|_\infty=1}|\int f \p\rho_\gamma|=\sup_{\|f\|_\infty=1}|\tr ( f [\p_{a_b},\gamma])|\ls \|[\p_{a_b},\gamma]\|_{I^{0,1}}$ 
 Now, writing $\gamma=\ka^2$ and combining $\p_{a_b}$ with one of the $\ka$'s, we estimate furthermore $\|[\p_{a_b},\gamma]\|_{I^{0,1}} \ls\|\ka\|_{I^{1,2}} \|\ka\|_{I^{0,2}}$. Then we interpolate between $s=0$ and $s=1$ to get the first inequality.

 For the second inequality, we 
 let $\gamma=\kappa^2$ and write $\rho_\gamma(x)=\int \kappa(x, y) \kappa(y, x)$. It is not hard to see that
\[\|\rho_\gamma\|_{W^{1, 1}} \ls \|\kappa\|_{I^{1, 2}}  \|\kappa\|_{I^{0, 2}}=( \|\gamma\|_{I^{1, 1}}  \tr \gamma)^{1/2}. \]

The second inequality, together with $\int_\Om \rho^q\le (\int_\Om \rho^{(q-v)/(1-v)})^{1-v}  (\int_\Om \rho)^{v}$,  $v< 1< q$, and a Sobolev inequality $\|\rho_\gamma\|_{L^p} \ls \|\rho_\gamma\|_{W^{s, 1}}^{1-v/q}, s> 2(1-1/p)$,
  implies \eqref{den-ineq3}

The first inequality, together with the Sobolev  inequality, $\|\rho_\g\|_{L^{3/2}}\ls \|\rho_\gamma\|_{W^{s, 1}}, s>2(1-1/q) $,  gives \eqref{den-ineq4}. \end{proof}

\begin{lemma} \label{lem:e2gam-est} We have for any $r\in (0,1)$,
\begin{align}\label{e2gam-est'''}
	0\le \Tr(|e|^2 \gamma)\ls  \|\g\|_{I^{1, 1}}^{1-r}\  (\tr \gamma)^{r}  \|e\|_{H^1}^2.  
\end{align}
 \end{lemma}
\begin{proof} We write $\Tr(|e|^2 \gamma)=\int_\Om |e|^2 \rho_\gamma $ and apply to this the H\"older and Sobolev inequalities to obtain
\[0\le \Tr(|e|^2 \gamma)\ls (\int_\Om |e|^{2p})^{1/p} (\int_\Om \rho_\gamma^q)^{1/q}\ls \|e\|_{H^1}^2\|\rho_\gamma\|_{L^{q}}, 
\]
where $1/p+ 1/q=1.$ This inequality together with  inequality  \eqref{den-ineq3} above gives  \eqref{e2gam-est'''} with $r\in (0,1)$. 
 \end{proof}
 We present another way to prove \eqref{e2gam-est'''} assuming the non-abelian interpolation inequality 
\begin{align}\label{e2gam-est''}
	 \|\ka \|_{I^{s, 2}} \ls  \|\ka\|_{I^{1, 2}}^s \|\ka\|_{I^{0, 2}}^{1/2-s}.\end{align}
 We use $\g=\ka \ka$ and write, for any $s, t>0$,
\begin{align}\label{e2gam-est''}
	0\le \Tr(|e|^2 \gamma)&= \Tr( M_b^{-s}|e|^2  M_b^{-t} M_b^{t}\ka \ka  M_b^{t})\\
	&\ls \| M_b^{-s}|e|^2 M_b^{-t}\|  \|M_b^{t}\ka \|_{I^{2}}  \|\ka M_b^{s}\|_{I^{2}} \end{align}
The last  inequality, together with \eqref{e2gam-est''} and relative bound \eqref{Sob-ineq3}, gives  \eqref{e2gam-est'''} with $r\in (0,1)$. 	
	

\section{Quasifree reduction} \label{sec:qf-reduct} 

In general, a many body evolution can be defined on states  (i.e. positive linear (`expectation') functionals) on the CAR or Weyl CCR algebra ${\frak W}$ over, say, Schwartz space $\cS(\R^d)$. Elements of this algebra are operators acting  on  the 
  fermionic/bosonic Fock space $\cF$, with annihilation and creation operators $\psi(x)$ and $\psi^*(x) $.\footnote{For a more detailed informal description, see \cite{BBCFS}.} 
 Given a quantum Hamiltonian $H$ on $\cF$, the evolution of states is  given by the 
von~Neumann-Landau equation \begin{align}\label{vNeum-eq} 
    i \partial_t\om_t(A) = \om_t([A, H]) \,,\ \forall A\in   \frak W.  
\end{align}
   (We leave out technical questions such as a definition of $\om_t([A, H])$ as $[A, H]$ is not in $\frak W$.) 

Let $N:=\int dx \; \psi^*(x) \psi(x)$ be the particle number operator. We distinguish between (a) {\it confined} systems with $\om(N)<\infty$, as in the case of BEC experiments in traps, and (b) {\it thermodynamic} systems with  $\om(N)=\infty$. In the former case the states are given by density operators on $\cF$, i.e. $\om (A)=\tr (A D)$, 
 where $D$ is a positive, trace-class operator with unit trace on $\cF$ (see e.g.  \cite{BLS}, Lemma 2.4). 

As the evolution \eqref{vNeum-eq} is practically intractable, 
one is interested in manageable approximations. The natural and most commonly used ones are one-body ones, which trade the number of degrees of freedom for the nonlinearity. 

The most general one-body approximation is given in terms of quasifree states. 
 A quasifree state $\qf$ determines and is determined by the truncated expectations to the second order:
\begin{equation} \label{omq-mom} \begin{cases}
    \phi(x):=\qf (\psi (x)),\\
    \gamma (x,y) :=   \qf[\psi^{*}(y) \, \psi(x)] - \qf [\psi^{*}(y)] \, \qf [\psi(x)] ,\\
    \al (x,y):=   \qf [\psi(x) \, \psi(y)] - \qf [\psi(x)] \, \qf [\psi(y)] \,.
\end{cases} \end{equation}
Let $\gamma$ and $ \al$ denote the operators with the integral kernels $\gamma (x,y)$ and $\s (x,y)$. After stripping off the spin components, this definition implies \eqref{gam-al-prop}. 
\DETAILS{that
\begin{equation} \label{gam-al-prop}\gamma=\gamma^*\ge 0\ \text{ and }\ \s^*=\bar\s, \end{equation} 
where $\bar\sigma =C \sigma C$ with $C$ being the complex conjugation. (For  confined systems, $\g$ and $\s$ are trace class and Hilbert-Schmidt operators, respectively, with $\tr(\g)<\infty$ giving  the particle number, while for thermodynamic systems, they are not, unless we place the system in a box.) The  operator  $\g$ can be considered as a one-particle density matrix of the system, so that $\tr \g =\int \g(x, x) dx$ is the total number of particles.  The  operator  $\s$  gives the coherence content of the state.} 

However, the property of being quasifree is not preserved by the dynamics  \eqref{vNeum-eq} and the main question here is how to project the true quantum evolution onto the class of quasifree  states. 

Following \cite{BBCFS}, we define self-consistent approximation as the restriction of the many-body dynamics to quasifree states.
More precisely, 
we map the solution $\om_t$ of \eqref{vNeum-eq}, with an initial state $\om_0$, 
 to the family $\qf_t$ 
 of quasifree states satisfying 
\begin{align}\label{eq-vNeum-quasifree}
    i  \partial_t\qf_t( A) = \qf_t([ A, H]) \,   
\end{align}
for all 
observables $A$, which are at most quadratic in the creation and annihilation operators,  with an initial state $\qf_0$, which is the quasifree projection of $\om_0$. 
We call this map the \textit{quasifree reduction} of equation \eqref{vNeum-eq}. 

 Of course, we cannot expect  $\qf_t$ to be a good approximation of $\om_t$, if $\om_0$ is far from the manifold of quasifree states.

Evaluating \eqref{eq-vNeum-quasifree}  for monomials $A\in 
 \{\psi(x),\psi^*(x)\psi(y),\psi(x)\psi(y)\},$ yields a system of coupled nonlinear
PDE's for  $(\phi_t,\gamma_t,\al_t)$. For the standard  any-body Hamiltonian, 
\begin{align}  \label{H} 
    H = \int dx \; \psi^*(x)h\psi(x) 
    + \frac{1}{2}\int dx dy \; v(x, y)\psi^*(x) \psi^*(y)
    \psi(x) \psi(y) \,,
\end{align}
with $h:= -\Delta +V(x)$ acting on the variable $x$ and $v$ a pair potential of the particle interaction, defined on Fock  space, $\cF$,
these give the (time-dependent) \textit{Bogolubov-\-de Gennes (BdG) or the (time-dependent) Hartree-\-Fock-\-Bogo\-lubov (HFB) equations}, depending on whether we deal with fermions or bosons. In the former case, the starting Hamiltonian is given by the BCS theory and one takes $\phi_t(x)=0$ and $v(x, y)$ is non-local, in the latter case, $v(x, y)\equiv v(x-y)$ is a local operator.



\begin{thebibliography}{}

\bibitem{AHS}  J. Avron, I. Herbst and B. Simon, Separation of the centre of mass in homogeneous magnetic fields, Ann Phys 114, 431-451 (1978)).

\bibitem{BBCFS} V. Bach, S.  Breteaux, Th. Chen, J. M. Fr\"ohlich  and I.M. Sigal, The time-dependent Hartree-Fock-Bogoliubov equations for bosons, arXiv 2018.

\bibitem{BLS} V. Bach, E. H. Lieb, and J. P. Solovej, Generalized Hartree-Fock theory
and the Hubbard model, J. Stat. Phys., 76 (1994), pp. 3--89. 

\bibitem{Brisl} C. Brislawn, Traceable integral kernels on countably generated measure spaces. Pacific Journal of Mathematics. Vol. 150, No 2, (1991), pp. 229 -- 240





\bibitem{Cyr} M. Cyrot, Ginzburg-Landau theory for superconductors, Reports on Progress in Physics, Volume 36, Number 2, 1973.

\bibitem{dG} P.G. de Gennes, 
 Superconductivity of Metals and Alloys (WA Benjamin, New York, 1966), Vol. 86.



\bibitem{FHSS} R. Frank, Ch. Hainzl, R. Seiringer, J.-P. Solovej, Microscopic derivation of the Ginzburg-Landau model, Journal AMS 25 (3), 667-713, 2012. 



\bibitem{GS}
S.~J.~Gustafson and I.~M. Sigal. 
\newblock {\em Mathematical Concepts of Quantum Mechanics}. The 2nd edition.
\newblock Universitext. Springer-Verlag, Berlin, 2011.




\bibitem{Ha} G. Br\"aunlich, C. Hainzl, and R. Seiringer, Translation-invariant quasi-free states for fermionic systems and the BCS approximation, Rev. Math. Phys., 26 (2014), p. 1450012. 

\bibitem{HHSS} Ch. Hainzl, Em. Hamza, R. Seiringer, J. P. Solovej, The BCS functional for general pair interactions, Commun. Math. Phys. 281, 349--367 (2008).



\bibitem{HaiSei} Ch. Hainzl and R. Seiringer, The BCS-functional of superconductivity and its mathematical properties JMP 57, Issue 2, 2016.




\bibitem{Legg}  A.J. Leggett: Diatomic molecules and Cooper pairs. Modern trends in the theory of condensed matter, J. Phys. (Paris) Colloq, C7--19 (1980). 

\bibitem{LewNamSchl}  M. Lewin, P. T. Nam, and B. Schlein. Fluctuations around Hartree states in the mean-field regime. Am. J. Math., 137(6):1613--1650, 2015. doi:10.1353/a jm.2015.0040.





\bibitem{L2} G. Lindblad. Expectations and entropy inequalities for finite quantum systems. Comm. Math. Phys. 39, 111--119 (1974)


\bibitem{MartRoth}  P.A. Martin and  F. Rothen,  Many-body Problems and Quantum Field Theory. 
Springer, 2004.






\bibitem{Sig0}   
  I.~M.~Sigal,   Magnetic vortices, Abrikosov lattices and automorphic functions, in Mathematical and Computational Modelling (with Applications in Natural and Social Sciences, Engineering, and the Arts), A John Wiley $\&$ Sons, Inc.,   2014.





\end{thebibliography}
\end{document}